%% file: main.tex
\documentclass[11pt]{article}
\usepackage
{
        amssymb,
        amsmath,
        amsthm,
        nicefrac,
        xcolor,
        fullpage,
        bbm,
}

\usepackage[unicode,colorlinks=true,citecolor=black,filecolor=black,linkcolor=black,urlcolor=black, pdfpagelabels,plainpages=false]{hyperref}

\usepackage{bookmark}
\newcommand {\ONE} {\mathbbm{1}}

\DeclareMathOperator {\Var}   {Var}

\DeclareMathOperator {\cost}  {cost}

\newcommand {\set}   [1] {\left\{ #1 \right\}}
\newcommand {\brc}   [1] {\left(#1\right)}
\newcommand {\Exp}       {\mathbb{E}}

\newcommand {\Prob}  [1] {\Pr \brc{#1 }}

\newcommand {\E}     [1] {\Exp\left[#1\right]}

\newcommand {\bbR}    {\mathbb{R}}

\newcommand {\calC}    {{\cal{C}}}

\newcommand {\calE}    {{\cal{E}}}

\newcommand{\eps}{\varepsilon}

\newtheorem{theorem}{Theorem}[section]
\newtheorem{lemma}[theorem]{Lemma}

\newtheorem{claim}[theorem]{Claim}
\newtheorem{corollary}[theorem]{Corollary}
\newtheorem{definition}[theorem]{Definition}
\newtheorem{observation}[theorem]{Observation}

\title{Performance of Johnson--Lindenstrauss Transform for\\$k$-Means and $k$-Medians Clustering}
\author{Konstantin Makarychev \\ Northwestern University \and Yury Makarychev\thanks{Supported by NSF CCF-1718820 and NSF Career CCF-1150062.} \\ TTIC \and Ilya Razenshteyn \\ Microsoft Research}
\date{}
\begin{document}
\maketitle
\begin{abstract}
Consider an instance of Euclidean $k$-means or $k$-medians clustering. We show that the cost of the optimal solution is preserved up to
a factor of $(1+\varepsilon)$ under a projection onto a random $O(\log(k / \eps) / \eps^2)$-dimensional subspace.
Further, the cost of \textit{every} clustering is preserved within $(1+\varepsilon)$.
More generally, our result applies to any dimension reduction map satisfying a mild sub-Gaussian-tail condition.
Our bound on the dimension is nearly optimal. Additionally, our result applies to Euclidean
$k$-clustering with the distances raised to the $p$-th power for any constant $p$.

For $k$-means, our result resolves an open problem posed by Cohen, Elder, Musco, Musco, and Persu (STOC 2015);
for $k$-medians, it answers a question raised by Kannan.
\end{abstract}
\input{intro.tex}
\input{proof-overview.tex}
\input{prelim.tex}

\input{bounding-cost.tex}
\input{main-theorems.tex}
\input{main-measure-lemma.tex}
\input{non-distorted-core.tex}
\input{kirszbraun.tex}
\section*{Acknowledgments}
We thank Piotr Indyk, Jerry Li, Jelani Nelson, and Tal Wagner for useful discussions.
\bibliographystyle{alpha}
\bibliography{bib}
\input{appendix.tex}
\input{tail-bound.tex}
\input{fast-jl.tex}
\end{document}

%% file: intro.tex
\section{Introduction}

The Euclidean \emph{$k$-clustering problem} with the $\ell_p$-objective is
defined as follows. Given a dataset $X \subset \bbR^m$ of $n$ points, the
goal is to find a partition $\calC = \{C_1, C_2, \ldots, C_k\}$ of $X$ into
$k$ parts (clusters) that minimizes the following cost function:
$$
\cost_p \calC = \sum_{i=1}^k \min_{u_i \in \bbR^m} \sum_{x \in C_i} \|x - u_i\|^p,
$$
where $\|\cdot\|$ from now on denotes the Euclidean ($\ell_2$) norm, and the optimal points $u_i$
are called \emph{centers} of clusters $C_i$. This
problem is a generalization of the $k$-median ($p = 1$) and the $k$-means
($p = 2$) clustering.
Algorithms for $k$-clustering (especially the Lloyd's heuristic~\cite{lloyd1982least} for $k$-means)
are used in virtually every area of data science (see~\cite{jain2010data} for a survey),
for data compression, quantization, and transmission over noisy channels~\cite{farvardin1990study},
and for hashing, sketching and similarity search~\cite{jegou2011product}.
In this paper we study data-oblivious \emph{dimension reduction} for $k$-clustering, which can
be used to speed up clustering algorithms. This line of work has been initiated by
Boutsidis, Zouzias, and Drineas~\cite{boutsidis2010random} and prior to this work, the best bounds
were due to Cohen, Elder, Musco, Musco and Persu~\cite{cohen2015dimensionality}.
Before stating our results, let us briefly recall the notion of Euclidean
dimension reduction (see~\cite{naor2018metric} for a broad overview of the area).

\paragraph{Dimension reduction.} The cornerstone
dimension reduction statement for the Euclidean distance is the Johnson--Lindenstrauss Lemma~\cite{JL84}.
For positive reals $p, q, \eps$, we write $p \approx_{1 + \eps} q$ if
$\frac{1}{1 + \eps} \cdot p \leq q \leq (1 + \eps) \cdot p$.
\begin{theorem}[\cite{JL84}]
\label{distributional_jl}
There exists a family of random linear maps $\pi_{m,d} \colon \bbR^m \to \bbR^d$ with the following properties.
For every $m \geq 1$, $\eps, \delta \in (0; 1/2)$ and all $x \in \bbR^m$, we have
\begin{equation*}
\Pr_{\pi \sim \pi_{m,d}}\left(\|\pi x\| \approx_{1+\eps} \|x\|\right) \geq 1 - \delta,
\end{equation*}
where $d = O\left(\frac{\log(1 / \delta)}{\eps^2}\right)$.
\end{theorem}
A straightforward corollary is that one is able to embed any $n$-point subset of a Euclidean space into an $O\left(\frac{\log n}{\eps^2}\right)$-dimensional space, while preserving all of the pairwise distances
up to $(1 + \eps)$. This bound is known to be tight~\cite{alon2003problems, larsen2017optimality}. The attractive feature of the dimension
reduction procedure given by Theorem~\ref{distributional_jl} is that it is \emph{data-oblivious} i.e., the distribution over linear maps is independent of
the set of points we apply it to.

There are several constructions of families of random maps $\pi_{m,d}$ that satisfy Theorem~\ref{distributional_jl}:
projections on a random subspace~\cite{JL84,dasgupta2003elementary} and maps given by matrices with i.i.d.\ Gaussian and sub-Gaussian entries~\cite{indyk1998approximate,achlioptas2003database,klartag2005empirical}.
All of these constructions satisfy a certain additional condition, which we will need later.

\begin{definition}
\label{def_standard}
A family of random linear maps $\pi_{m,d} \colon \bbR^m \to \bbR^d$
is called \emph{sub-Gaussian-tailed} if for every unit vector $x \in \bbR^m$ and every $t > 0$, one has:
$$
\Pr_{\pi \sim \pi_{m,d}}(\|\pi x\| \geq 1 + t) \leq e^{-\Omega(t^2 d)}.
$$
\end{definition}

\paragraph{Our result.}
Consider an instance of Euclidean $k$-clustering. We show that its cost is preserved up to
a factor of $(1+\varepsilon)$ under dimension-reduction projection into an $O(\log(k / \eps) / \eps^2)$-dimensional space.
Further, the cost of \textit{every} clustering is preserved within a factor of $(1+\varepsilon)$.
Our result applies to dimension reductions based on orthogonal and Gaussian projections and, more generally,
any dimension reductions satisfying the sub-Gaussian-tail condition (in the sense of Definition~\ref{def_standard}).
Our bound on the dimension is nearly optimal.

For $k$-means, our result resolves an open problem posed by Cohen, Elder, Musco, Musco, and Persu~\cite{cohen2015dimensionality}.
For a partition $\calC = (C_1, C_2, \ldots, C_k)$
of $X \subset \bbR^m$ and a linear map $\pi \colon \bbR^m \to \bbR^d$,
we denote by $\pi(\calC)$ the respective partition of the image of $X$ under $\pi$.
We now state our main result formally.
\begin{theorem}
\label{main_th_informal}
Consider any family of random maps $\pi_{m,d} \colon \bbR^m \to \bbR^d$ that satisfies Theorem~\ref{distributional_jl}
and is sub-Gaussian-tailed (satisfies Definition~\ref{def_standard}).
Then for every $m \geq 1$, $\eps, \delta \in (0; 1/4)$ and $p \geq 1$,
the following holds. For every finite $X \subset \bbR^m$ we have
$$
\Pr_{\pi \sim \pi_{m,d}}\left(\cost_p \calC \approx_{1+\eps} \cost_p \pi(\calC) \mbox{ for all partitions } \calC = (C_1, C_2, \ldots, C_k) \mbox{ of $X$} \right) \geq 1 - \delta,
$$
where
\begin{equation}
\label{dim_bound}
d = O\left(\frac{p^4 \cdot \log \frac{k}{\eps \delta}}{\eps^2}\right).
\end{equation}
\end{theorem}
In fact, we show that Theorem~\ref{main_th_informal} holds under a slightly more general definition of a
\emph{standard} dimension reduction map (Definition~\ref{def:dim-red-parameters}), which is implied by sub-Gaussian tails.

Theorem~\ref{main_th_informal} readily implies that if one solves the $k$-clustering problem after dimension reduction with approximation $\lambda \geq 1$, then the same solution yields approximation
$(1 + O(\eps)) \cdot \lambda$ for the original instance.
It is almost
immediate to show the guarantee given by Theorem~\ref{main_th_informal} for a \emph{fixed}
partition~$\calC$, however the total number of partitions is exponential,
and we cannot afford to take the union bound over them. In fact, the actual proof of
Theorem~\ref{main_th_informal}
is much more delicate.

For $k$-means ($p = 2$), the dimension bound $O(\log n / \eps^2)$ easily follows from
the Johnson--Lindenstrauss lemma (Theorem~\ref{distributional_jl})
and the fact that one can express the cost function of $k$-means via pairwise
distances. However, in many applications of clustering, $k \ll n$, in which case the bound we obtain is much stronger. For $p \ne 2$, even the weaker
$O(\log n)$ bound was not known before.

There are several known dimension reduction maps that essentially satisfy a variant of
Theorem~\ref{distributional_jl} (sometimes with somewhat worse bound on $d$)
but do not necessarily satisfy Definition~\ref{def_standard}: most notably, ``sparse'' constructions~\cite{dasgupta2010sparse,kane2014sparser} as well as ``fast'' constructions based on subsampled randomized Fourier-like transforms~\cite{ailon2006approximate,ailon2009fast,krahmer2011new,ailon2013almost,nelson2014new}.
To demonstrate versatility of our techniques, in Appendix~\ref{fast_jl_appendix} we show how to adjust the proof of Theorem~\ref{main_th_informal} to handle a variant of the ``fast'' map introduced in~\cite{ailon2006approximate} (yielding, roughly speaking, a quadratically worse bound on the target dimension $d$
than that given by Theorem~\ref{main_th_informal}).

Finally, let us note that the bound~(\ref{dim_bound}) is essentially optimal in the following two ways. First, if we
want to preserve the cost of all the $k$-clusterings of a dataset, then one needs at least $\Omega(\log k / \eps^2)$
dimensions for \emph{any} (possibly non-linear) dimension reduction map. Indeed, if we consider a dataset consisting of
$k + 1$ points, then preserving the cost of all $k$-clusterings is equivalent to preserving all the pairwise distances.
The desired lower bound then follows from the main result of~\cite{larsen2017optimality}. Second, if we use a Gaussian
matrix for dimension reduction, and only care about preserving the \emph{optimal} clustering, we still need
$\Omega(\log k / \eps^2)$ dimensions. Indeed, consider a set of $t$ pairs of points such that the distance within each
pair except one is $1$, the distance within the remaining one pair is $1 - C \eps$, where $C$ is a large enough
constant, and the pairs are very far apart from each other. Suppose that $k = 2t - 1$. Then the only approximately
optimal clustering consists of the ``special'' pair and the remaining points as singletons. But if we reduce dimension
using a Gaussian matrix to much fewer than $\log k / \eps^2$ dimensions, it is likely that some other pair will become
noticeably closer than the special pair. This changes the optimal clustering showing the desired lower bound.

\begin{figure}
    \centering
    \begin{tabular}{|c|c|c|}
    \hline
         Reference & Dimension $d$ & Distortion  \\
         \hline
         Folklore & $O(\log n / \eps^2)$ & $1 + \eps$\\
         \cite{boutsidis2010random} & $O(k / \eps^2)$ & $2 + \eps$\\
         \cite{cohen2015dimensionality} & $O(k / \eps^2)$ & $1 + \eps$ \\
         \cite{cohen2015dimensionality} & $O(\log k / \eps^2)$ & $9 + \eps$\\
         This work & $O(\log(k / \eps) / \eps^2)$ & $1 + \eps$\\
                  \hline
    \end{tabular}
    \caption{Data-oblivious dimension reduction for $k$-means}
    \label{prior_work_k_means}
\end{figure}

\paragraph{Related work.}
There is a large body of literature on dimension reduction for $k$-means
(which corresponds to $p = 2$).
Within this line of work, there are two kinds of results: data-oblivious and data-dependent.
Data-oblivious results provide guarantees qualitatively similar to our
Theorem~\ref{main_th_informal}
and are summarized in Figure~\ref{prior_work_k_means}.
Let us now give a brief overview.

As mentioned previously, the bound on the dimension $O(\log n / \eps^2)$ is a
simple application of Theorem~\ref{distributional_jl}. The first bound
independent of $n$ was obtained by Boutsidis, Zouzias and
Drineas~\cite{boutsidis2010random}, who showed that $O(k / \eps^2)$
dimensions are enough for distortion $2 + \eps$. The best bounds prior to the
present work are due to Cohen, Elder, Musco, Musco and
Persu~\cite{cohen2015dimensionality}. They showed two incomparable bounds:
$O(k / \eps^2)$ dimensions with distortion $1 + \eps$, and $O(\log k /
\eps^2)$ dimensions with distortion $9 + \eps$. The $9 + \eps$ bound on the
distortion follows from all pairwise distances between $k$ optimal
centers being approximately preserved. Our Theorem~\ref{main_th_informal}
improves upon both of the bounds shown in~\cite{cohen2015dimensionality}: we
get $O(\log(k / \eps) / \eps^2)$ dimensions with distortion $1 + \eps$, thus
resolving an open problem posed in~\cite{cohen2015dimensionality}.

The literature on data-dependent dimension reduction for $k$-means is ample \cite{drineas1999clustering, sarlos2006improved, boutsidis2009unsupervised, boutsidis2010random, feldman2013turning, boutsidis2013deterministic, boutsidis2015randomized, cohen2015dimensionality} and we refer the reader
to~\cite{cohen2015dimensionality} for a~comprehensive overview.
Let us note that none of these results obtain dimension better than $k$.

For $p \ne 2$, even the $\log n$ bound was not known previously. The question of obtaining the
$\log n$ bound for $p = 1$ ($k$-median) was explicitly posed recently by Kannan~\cite{kannan2018}.

A notion related to dimension reduction is that of a \emph{coreset}, which is a small subsample of a~dataset that approximately preserves
the cost of $k$-clustering. A good overview of this rich line of work appears in~\cite{sohler2018strong}.

Independently of our work, Becchetti, Bury, Cohen-Addad, Grandoni and Schwiegelshohn~\cite{becchetti2019} obtained a
result for $k$-means under dimension reduction. Their result applies only to $k$-means ($p = 2$); their bound on the
target dimension $O\left(\frac{\log k + \log \log n}{\eps^6} \cdot \log \frac{1}{\eps}\right)$ is somewhat weaker than
ours and depends on $n$.

%% file: proof-overview.tex
\subsection{Proof Overview}

As we discussed earlier, it is easy to show that for any fixed clustering $\calC =(C_1,\dots,C_k)$, we have $\cost_p(\pi \calC) \approx_{1+\varepsilon} \cost_p(\calC)$
w.h.p. In particular, for the optimal clustering $\calC^*$, $\cost_p(\pi \calC^*) \approx_{1+\varepsilon} \cost_p(\calC^*)$, and, consequently,
the cost of the optimal clustering for $\pi X$ is upper bounded by the cost of the optimal clustering for $X$ up to a factor of $(1+\varepsilon)$ w.h.p. However,
it is not at all obvious how to obtain a lower bound on the cost of the optimal solution for $\pi X$  since there may exist a~clustering $\pi \calC'$ of $\pi X$ which
is better than $\pi \calC^*$. Note that we cannot use the union bound to prove Theorem~\ref{main_th_informal} as the number of possible clusterings of $X$ is
exponential in $n$, but the dimension $d$ does not depend on $n$.

In this section, we discuss the main ideas we use in the proof of Theorem~\ref{main_th_informal}. We show that with high probability the following
two statements hold
(a) for all $\calC$ we have $\cost_p(\calC) \leq (1+\varepsilon) \cost_p(\pi \calC)$
and
(b) for all $\calC$ we have $\cost_p(\pi \calC) \leq (1+\varepsilon) \cost_p(\calC)$.
The proofs of these statements are similar. To simplify the exposition, we focus on the former inequality in this proof overview.

To illustrate our approach, first consider an easy case when $X$ is embedded into a $d=O(\log n/\varepsilon^2)$
dimensional space. In this case, all distances between points in $X$ are approximately preserved w.h.p. That is, $\|x'-x''\|\approx_{1+\varepsilon}\|\pi x'-\pi x''\|$
for all $x',x''\in X$. We prove that if all distances in $X$ are approximately preserved then for every clustering $\calC=(C_1,\dots, C_k)$ we have
$\cost_p(\calC) \approx_{1+\varepsilon} \cost_p(\pi \calC)$. In fact, we prove a slightly stronger statement:
for every cluster $C\subset X$ we have $\cost_p(C) \approx_{1+\varepsilon} \cost_p(\pi C)$. As we shall see in a moment, the proof
is immediate for $k$-means but requires some work for $k$-median and other $\ell_p$ objectives.

For the $k$-means objective ($p=2$), we can use the following
well-known formula:
\begin{equation}\label{eq:intro:kmeans-costC}
\cost_2(C) = \frac{1}{2|C|} \sum_{(x',x'')\in C\times C} \|x'-x''\|^2,
\end{equation}
and, similarly,
\begin{equation}\label{eq:intro:kmeans-cost-pi-C}
\cost_2(\pi C) = \frac{1}{2|C|} \sum_{(x',x'')\in C\times C} \|\pi x'- \pi x''\|^2.
\end{equation}
Since each term $\|x'-x''\|^2$ in (\ref{eq:intro:kmeans-costC}) approximately equals the corresponding term $\|\pi x'- \pi x''\|^2$ in (\ref{eq:intro:kmeans-cost-pi-C}), we have
$\cost_2(C) \approx_{1+O(\varepsilon)} \cost_2(\pi C)$.

\medskip

\noindent\textbf{One-point robust extension.}
The above proof does not generalize to $\ell_p$ objectives with $p\neq 2$. So our proof relies on the Kirszbraun
theorem~\cite{kirszbraun1934zusammenziehende}.

\begin{theorem}[Kirszbraun theorem]
For every subset $X\subset \bbR^d$ and $L$-Lipschitz map\footnote{Recall that $\varphi$ is an $L$-Lipschitz map if
for all $x',x''\in X$, we have $\|\varphi(x')-\varphi(x'')\| \leq L\|x'-x''\|$.}  $\varphi: X \to \bbR^m$, there exists an $L$-Lipschitz
extension $\widetilde\varphi$ of $\varphi$ from $X$ to the entire space $\bbR^d$.
\end{theorem}

Let $Y=\pi(C)$. Observe that the map $\pi: C \to \bbR^d$ and inverse map $\pi^{-1}: Y \to \bbR^m$ are $(1+\varepsilon)$-Lipschitz with high probability.
Let $u$ be the optimal center for cluster $\pi C$. Using the Kirszbraun theorem, we extend the map $\pi^{-1}: Y \to \bbR^m$ to $\widetilde{\varphi}: \bbR^d \to \bbR^m$ and then lift the optimal center $u$ from $\bbR^d$ to $\bbR^m$ by letting $v = \widetilde{\varphi}(u)$. Then, for all
$y\in Y$ we have $\|v - \pi^{-1} y\| \leq (1+\varepsilon) \|u - y\|$ or, equivalently,
for all $x\in C$ we have $\|v - x\| \leq (1+\varepsilon) \|u - \pi x\|$.
We pick point $v$ as the center for the cluster $C$ and obtain the following bound:
$$\cost_p(C) \leq \sum_{x\in C} \|x - v\|^p \leq \sum_{x\in C} (1+\varepsilon)^p \|\pi x - u\|^p = (1+\varepsilon)^p\cost_p(\pi C).$$

\medskip

We now return to the case when $d=c_p \log (k/(\delta \varepsilon))/\varepsilon^2)$ (where $c_p$ only depends on $p$). Observe that when we
reduce the dimension of $X$ to $d$, while most pairwise distances in $X$ are approximately preserved, some are distorted. The Kirszbraun theorem
does not hold in this setting. We prove a robust 1-point extension theorem (Theorem~\ref{thm:robust-Kirszbraun}).

Loosely speaking, this theorem states the following. Consider a~finite set $C\subset \bbR^d$ and map $\varphi: C\to \bbR^m$, satisfying the following condition ($\star$):
\begin{itemize}
  \item for every $x\in C$, the distance from $x$ to all but a $\theta$ fraction of $x'\in C$ is $(1+\varepsilon)$-preserved under~$\varphi$.
\end{itemize}
Then, for every point $u \in \bbR^d$, there exists a point $v\in \bbR^m$ such that for all but $\theta'$ fraction of points $x \in C$, we have $\|x - u\| \leq (1+\varepsilon') \|\varphi(x)- v\|$ where $\varepsilon'=O(\varepsilon)$ and $\theta' =  O(\theta/\varepsilon)$.

\medskip

\noindent\textbf{Worst cluster for each projection.} Consider a random
dimension reduction map $\pi \colon \bbR^m \to \bbR^d$. We prove that for
every cluster $C$,
\begin{equation}
\cost_p C \leq (1+\varepsilon') \cost_p \pi(C) + \varepsilon'' \cost_p \calC^* ,\label{eq:bound-main-2-overview}
\end{equation}
where $\cost_p \calC^*$ is the cost of the optimal $k$-means clustering. Here, $\varepsilon'$ and $\varepsilon''$ are
some parameters (see~(\ref{eq:bound-main-2}) for details). For each realization of $\pi$, let us pick
a subset $C\subset X$ for which the gap between $\cost_p(C)$ and $\cost_p(\pi C)$ is the largest. We would
like to use our new 1-point extension theorem to lift the optimal center of $\pi C$ from $\bbR^d$ to
the original space $\bbR^m$. However, we cannot do this directly, since it may be the case that most pairwise
distances in $C$ are distorted.

\medskip

\noindent\textbf{Everywhere sparse distortion graphs.} To deal with this
problem, we define a distortion graph $G$ for map $\pi: X\to \bbR^m$. The
vertex set of $G$ is $X$. Two points $u,v\in X$ are joined by an edge in $G$
if $\pi$ distorts the distance between $u$ and $v$ by a factor of at least
$(1+\varepsilon)$. We say that a subgraph $H$ of $G$ is $\theta$-everywhere
sparse if the degree of all vertices in $H$ are upper bounded by $\theta
|V_H|$. Observe that the probability that $(u,v)$ is an edge in $G$ is less
than $1/poly(k)$ if $d\gtrsim \log (k) /\varepsilon^2$. Note that the
distortion graphs for $\pi$ and its inverse $\pi^{-1}$ are the same.

We can now restate condition ($\star$) of the 1-point extension theorem as
follows: The distortion graph for the map $\varphi$ is $\theta$-everywhere
sparse. Note, however, that the induced distortion graph $G[C]$ for maps
$\pi$ and $\pi^{-1}$ is not necessarily $\theta$-everywhere sparse.
Nevertheless, we show that it is possible to remove some outlier points from
$X$ to make the induced graph $G[C]$ sparse. Specifically, we prove that
there exists a random set $B$ (which depends on $\pi$) such that 
\begin{itemize}
  \item[(a)] the graph $G[C\setminus B]$ is $\theta$-everywhere sparse;
  \item[(b)] $\Pr(x\in B)\leq \theta$ for all $x\in B$.
\end{itemize}

The proof of the existence of $B$ is fairly simple for the special case when the size of cluster $C$ is approximately $n/k$.
Mark a point $x\in X$ as an outlier if its degree in $G$ is greater than $\theta n/(2k)$ and let $B$ be the set of all outliers.
Observe that the degree of every vertex in $C\setminus B$ is at most $\theta n/(2k)$ and, thus, the maximum degree in the
induced graph $G[C\setminus B]$ is at most $\theta n/(2k)$.

The expected degree of a vertex in $G$ is much smaller than $\theta^2 n/k^2$,
since the probability that the distance between two given points is distorted
is much less than $\theta^2/k^2$ (if $d\gtrsim \log (k/\theta)
/\varepsilon^2$). Hence, by Markov's inequality, $\Pr(x\in B)\ll
\theta/k$ for any $x\in X$ and, therefore, condition (b) is satisfied. Also,
by linearity of expectation, the expected size of $B$ is at most $\theta
n/k$, and by Markov's inequality, $|B|\ll n/k$ w.h.p. Thus, $|C\setminus
B|\geq n/(2k)$ w.h.p. and, consequently, $G[C\setminus B]$ is
$\theta$-everywhere sparse w.h.p. The proof for the general case when the
size of the cluster $C$ may be arbitrary is more involved (see
Theorem~\ref{thm:everywhere-sparse}). Roughly speaking, we reduce the general
case to the case when $|C|\approx n/k$ by introducing a carefully chosen
measure $\mu$ on $X$ that we use as a proxy for the set sizes.

\medskip

\noindent\textbf{Outliers.} We are now almost done. The distortion graph for
the set $C\setminus B$ is $\theta$-everywhere sparse. Thus, by the robust
1-point extension theorem applied to the set $\pi(C\setminus B)$ and map
$\pi^{-1}$, we can lift the optimal center of $C\setminus B$ to the space
$\bbR^m$. To finish the proof, it only remains to take care of two sets of
outliers.

The first set is the set
$B$. We slightly modify our data set. We move each outlier $x$ to the optimal
center of the cluster $\calC^*(x)$ that the point $x$ is assigned to in the
optimal solution to $k$-clustering and then
return $x$ back to its cluster $C$. Since the probability that $x$ is an
outlier is very small, the cost of moving $x$ is also small (see
Theorem~\ref{thm:cost-one-cluster}).

The second type of outliers are the points $x\in C$ such that $\|x-v\|
> (1+\varepsilon) \|\pi x - u\|$, where $u$ is the optimal center of $C$ and
$v$ is the optimal center lifted to $\bbR^m$. The fraction of such outliers
among all points in $C$ is at most $\theta'$. We charge the connection costs
of these outliers (i.e., the costs of assigning outliers to the centers of
their clusters) to the connection costs of other vertices.

This concludes the proof. In the rest of the paper, we give a detailed proof of~Theorem~\ref{main_th_informal}.

%% file: prelim.tex
\section{Preliminaries}
Let $X \subset \bbR^d$ be an instance of Euclidean $k$-clustering with the $\ell_p$ objective.
We denote the optimal $k$-clustering of $X$ by $\calC^* = \{C_1^*,\dots, C_k^*\}$ and
its optimal centers by $c_1^*,\dots, c_k^*$. Given a cluster $S\subset X$, let
$\cost_p S = \min_{c\in \bbR^d} \sum_{u\in \bbR^d} \|u-c\|^p$ be its cost.
Given a clustering $\calC = \{C_1,\dots, C_k\}$, let $\cost_p \calC = \sum_{i=1}^k \cost_p C_i$
be its cost. In particular, $\cost \calC^*$ is the cost of the optimal clustering of $X$.
Given a map $\pi$ and $\calC = \{C_1,\dots, C_k\}$, denote by $\pi(\calC)$ the clustering
$\pi(C_1),\dots, \pi(C_k)$ of $\pi(X)$.

We denote the indicator of an event $\calE$ by $\ONE\set{\calE}$.

\begin{definition}\label{def:dim-red-parameters}
Random map $\pi:\bbR^m \to \bbR^d$ is an $(\varepsilon,\delta)$-random dimension reduction if
$$
    \frac{1}{1+\eps} \|x-y\| \leq \|\pi(x) - \pi(y)\| \leq (1+\eps)\|x-y\|
$$
with probability at least $1-\delta$ for every $x,y\in \bbR^m$.
Given $p\in[1,\infty)$, random map $\pi$ is an $(\varepsilon,\delta, \rho)$-dimension reduction  if it
additionally satisfies the following property
\begin{equation}
\E{\ONE\set{\|\pi(x)-\pi(y)\| > (1+\eps) \|x-y\|}
\Bigl(\frac{\|\pi(x)-\pi(y)\|^p}{\|x-y\|^p} -(1+\eps)^p \Bigr) } \leq \rho.\label{eq:tail-bound}
\end{equation}
Given $\varepsilon > 0$, we say that a random dimension reduction $\pi:\bbR^m \to \bbR^d$ is standard if it has
parameters $(\varepsilon,\delta, \rho)$ and $\delta \leq \exp(-c \eps^2 d)$,  $\rho \leq \exp(-c \eps^2 d)$
when $d > c' p/\eps^2$ for some constants $c,c'>0$.
\end{definition}
\begin{definition}
We say that $\pi$ approximately preserves or $(1+\eps)$-preserves the distance between $x$ and $y$ if $\frac{1}{1+\varepsilon} \|x-y\| \leq
\|\pi(x) - \pi(y)\| \leq (1+\varepsilon) \|x-y\|$. Otherwise, we say that $\pi$ distorts the distance between $x$ and $y$. For
brevity, we also say that $(x,y)$ is $(1+\eps)$-preserved in the former case and distorted in the latter case.

We define the distortion graph $G=(X, E)$ for $\pi$ as follows. Two points $x, y \in X$ are connected with an edge if $\pi$ distorts the
distance between them. Note that when $\pi$ is an $(\varepsilon,\delta)$-random dimension reduction, $G$ is a random graph and
$\Pr((x,y)\in E) \leq \delta$ for every $x,y \in X$.
\end{definition}
As we prove in Lemma~\ref{lem:sub-gauss-is-standard}, every dimension reduction that (1) satisfies Theorem~\ref{distributional_jl} and (2) is
sub-Gaussian tailed is standard. We note that all dimension reduction constructions (satisfying Theorem~\ref{distributional_jl}) that we are aware of are sub-Gaussian tailed and thus standard. In particular, we prove in Claim~\ref{claim:gaussian-is-standard} that the Gaussian dimension reduction is sub-Gaussian tailed.

%% file: bounding-cost.tex
\section{Dimension reduction preserves cluster costs}
In this section, we prove Theorem~\ref{main_th_informal}, the main result of this paper. The proof relies on Theorems~\ref{thm:everywhere-sparse} and~\ref{cor:main-Kriszbraun-corollary}, which we prove later in
Sections~\ref{sec:everywhere-sparse} and~\ref{sec:kirszbraun}. We note that Theorem~\ref{thm:everywhere-sparse} is a purely combinatorial theorem about random graphs, which does not deal with distances or embeddings.
On the other hand, Theorem~\ref{cor:main-Kriszbraun-corollary} is about a deterministic dimension reduction map $\varphi$, it states that $\varphi$ approximately preserves the cost of a cluster $C$ loosely speaking when the maximum degree in $G[C]$ (the graph induced by $C$ on the distortion graph) is much smaller than $|C|$.

\subsection{Theorems~\ref{thm:everywhere-sparse} and \ref{cor:main-Kriszbraun-corollary}}
We formally state Theorems \ref{thm:everywhere-sparse} and \ref{cor:main-Kriszbraun-corollary}.
\begin{definition}
A graph $H = (V, E)$ is $\theta$-everywhere sparse if $\deg u \leq \theta
|V|$ for every $u \in V$. (We allow $V$ to be empty.)
\end{definition}
\begin{theorem}\label{thm:everywhere-sparse}
Consider a finite set $X$ and a random graph $H = (V, E)$, where $V$ is a
random subset of $X$ and $E$ is a random set of edges between vertices in $V$
(we do not make any independence assumptions or any other implicit
assumptions about the distribution of $V$ and $E$). Let $\theta \in
(0,1/2)$. Assume that $\Prob{(x,y) \in E} \leq \delta$ for every $x,y\in X$,
where $\delta \leq \theta^7/600$ (if $x\notin V$ or $y\notin V$, then
$(x,y)\notin E$). Then there exists a random subset $V' \subset V$ ($V'$ is
defined on the same probability space as $H$; in other words, it is jointly
distributed with $H$) such that
\begin{itemize}
\item $H[V']$ is $\theta$-everywhere sparse,
\item $\Prob{u\in V\setminus V'} \leq \theta$ for all $u\in X$.
\end{itemize}
\end{theorem}

\begin{theorem}\label{cor:main-Kriszbraun-corollary}
Consider a finite multiset of points $\tilde C\subset \bbR^{d'}$ and a map
$\varphi: \tilde C\to \bbR^{d''}$. Assume that the distortion graph $G$ for
$\tilde C$ is $\theta$-sparse, where $\theta \leq 1/10^{p+1}$.
Then, for every $p\geq 1$ and $D = (1+\eps)^p(1+3^{p+2}\theta^{1/(p+1)})$, we
have
$$
D^{-1} \cost_p (\tilde C) \leq \cost_p(\varphi(\tilde C)) \leq D \cost_p(\tilde C).
$$
\end{theorem}

\subsection{Cluster costs are approximately preserved}
Assume that $C \subset X$ is a random subset/cluster of $X$; $C$ may depend
on the random projection $\pi$. We prove that $\cost_p C$ is approximately
equal to $\cost_p \pi(C)$ up to some multiplicative and additive
errors with high probability.

\begin{theorem}\label{thm:cost-one-cluster}
Consider an instance $X$ of Euclidean $k$-clustering with the $\ell_p$-objective.
Let $\pi$ be a $(\varepsilon, \delta, \theta)$-random dimension reduction, and $C$ be a random subset/cluster of $X$ (which may depend on $\pi$).
Assume that $\theta \leq 1/10^{p+1}$ and $\delta \leq \min(\theta^7/600, \theta/k)$. 
Denote the optimal clustering of $X$ by $\calC^*$. Then
\begin{align}
\cost_p \pi(C) &\leq A\left(\cost_p C + c_{\alpha\epsilon\theta} \cost_p \calC^*\right) \label{eq:bound-main-1}\\
\cost_p C &\leq A\left(\cost_p \pi(C) + c_{\alpha\epsilon\theta} \cost_p \calC^* \right) \label{eq:bound-main-2}
\end{align}
with probability at least $1 - \alpha-\binom{k}{2} \delta$, where $A = (1+\varepsilon)^{3p-2}
(1+3^{p+2}\theta^{1/(p+1)})$ and $c_{\alpha\epsilon\theta} =\frac{5(1+\eps)^p\theta}{\alpha\eps^{p-1}}$.
\end{theorem}
\begin{proof}
Let $c$ be the optimal center for cluster $C$. Denote the clusters of $\calC^*$ by $C_1^*,\dots, C_k^*$ and the optimal
centers for $C_1^*,\dots, C_k^*$ by $c_1^*,\dots, c_k^*$. Let $c^*(x) = c_i^*$ if $x\in C_i^*$. Consider the event $\calE$ that all the
distances between the centers $c_i^*$ are $(1+\eps)$-preserved. Note that $\Prob{\calE} \geq 1-\binom{k}{2} \delta$. We
assume below that $\calE$ happens. Let
$$C^{\circ} = \{x\in C: \pi \text{ approximately preserves the distance between } x \text{ and each } c_i^*\}.$$
Note that $\Prob{x\in C\setminus C^{\circ}} \leq \Prob{\text{the distance
between } x \text{ and some } c_i \text{ is distorted}} \leq k \delta$ for
all $x\in X$. We apply Theorem~\ref{thm:everywhere-sparse} to $H =
G[C^{\circ}]$ (the graph induced by $C^{\circ}$ on the distortion graph $G$).
We get a subset $C'\subset C^{\circ}$ such that
$G[C']$ is $\theta$-everywhere sparse and $\Prob{x\in C^{\circ}\setminus C'}
\leq \theta$. Observe that
$$\Prob{x\in C\setminus C'} \leq \Prob{x\in C\setminus C^{\circ}} + \Prob{x\in C^{\circ}\setminus C'} \leq k\delta + \theta \leq 2\theta.$$

Let
$$g(x) = \begin{cases}x, &\text{if } x\in C'\\ c^*(x), &\text{if } x\notin C'.\end{cases}$$ Consider multiset
$\widetilde C = g(C)$. By construction, every point in $\widetilde C$ is
either in $C'$ or equal to some $c_i^*$. Let $\tilde c$ be the optimal center
for $\widetilde C$. Since $G[C']$ is $\theta$-everywhere sparse, the map
$\pi$ approximately preserves the distances from every $x\in C'$ to at least
a $(1-\theta)$ fraction of the points in $C'$ and to all $c_i^*$. Further,
since $\cal E$ happens, all distances between centers $c_i^*$ are
$(1+\eps)$-preserved. Therefore, we can apply
Theorem~\ref{cor:main-Kriszbraun-corollary} to multiset $\widetilde C$ and
map $\varphi = \pi$. We get that
\begin{equation}\label{eq:bound-one-cluster}
D^{-1} \cost_p (\widetilde C) \leq \cost_p(\pi(\widetilde C)) \leq D \cost_p(\widetilde C)
\end{equation}
for $D = (1+\eps)^p(1+3^{p+2}\theta^{1/(p+1)})$. 
Note that
$$\cost_p C \leq \sum_{x\in C} \|x-\tilde c\|^p \quad\text{and}\quad \cost_p \widetilde C = \sum_{x\in C} \|g(x)-\tilde c\|^p,$$
since $\tilde c$ is the optimal center for $\widetilde{C}$.
Compare the summations in the upper bound for $\cost_p C$ and the formula for $\cost_p \widetilde C$ term by term.
For $x\in C'$, $g(x) = x$ and thus the terms in both summations are equal. For $x\notin C'$, the terms in the first and
second summations are equal to $\|x - \tilde c\|^p$ and $\|c^*(x) - \tilde c\|^p$, respectively.
Therefore,
$$\cost_p C - (1+\eps)^{p-1} \cost_p \widetilde C \leq
\sum_{x\in C\setminus C'} \|x - \tilde c\|^p - (1+\eps)^{p-1}\|c^*(x) -
\tilde c\|^p.$$ Observe that $\|x - \tilde c\| \leq \|x - c^*(x)\| + \|c^*(x)
- \tilde c\|$. Applying Lemma~\ref{ineq:sum-p-pow} with $r=1$ (see
Appendix~\ref{sec:ineq:sum-p-pow}), we get
$$\|x - \tilde c\|^p \leq (1+\eps)^{p-1}  \|c^*(x) - \tilde c\|^p + \left(\frac{1+\eps}{\eps}\right)^{p-1} \|x - c^*(x)\|^p.$$
Therefore,
$$\cost_p C - (1+\eps)^{p-1} \cost_p \widetilde C \leq \left(\frac{1+\eps}{\eps}\right)^{p-1} \sum_{x\in C\setminus C'} \|x - c^*(x)\|^p.$$
Similarly, we prove that 
\begin{align*}
\cost_p \tilde C - (1+\eps)^{p-1} \cost_p C &\leq \left(\frac{1+\eps}{\eps}\right)^{p-1} \sum_{x\in C\setminus C'} \|x - c^*(x)\|^p\\
\cost_p \pi(C) - (1+\eps)^{p-1} \cost_p \pi(\tilde C) &\leq \left(\frac{1+\eps}{\eps}\right)^{p-1} \sum_{x\in C\setminus C'} \|\pi(x) - \pi(c^*(x))\|^p\\
\cost_p \pi(\tilde C) - (1+\eps)^{p-1} \cost_p \pi(C) &\leq \left(\frac{1+\eps}{\eps}\right)^{p-1} \sum_{x\in C\setminus C'} \|\pi(x) - \pi(c^*(x))\|^p
\end{align*}
Combining these bounds with inequality~(\ref{eq:bound-one-cluster}), we
obtain 
\begin{align*}
\cost_p \pi(C) &\leq A\left(\cost_p C + \varepsilon^{1-p} \sum_{c\in X \setminus X'} R_x\right) \\
\cost_p C &\leq A\left(\cost_p \pi(\calC) + \varepsilon^{1-p} \sum_{c\in X \setminus X'} R_x\right)\\
\text {where } \quad A &= (1+\varepsilon)^{2(p-1)} D \quad\text{and}\quad
R_x = \|x - c^*(x)\|^p + \|\pi(c^*(x)) - \pi(x)\|^p.
\end{align*}

It remains to prove that $\varepsilon^{1-p} \sum_{x\in C\setminus C'}\ONE\set{\calE} R_x \leq c_{\alpha\epsilon\theta}
\cost_p \calC^*$ with probability at least $1-\alpha$. To this end, we show that $\E{\ONE\set{\calE}  \sum_{x\in C\setminus
C'} R_x} \leq 5(1+\eps)^p\theta \cost_p \calC^*$ and then use Markov's inequality. From property~(\ref{eq:tail-bound}) in
Definition~\ref{def:dim-red-parameters} we get that for every $x\in C$,
$$\E{\max(\|\pi(c^*(x)) - \pi(x)\|^p - (1 + \eps)^{p} \|c^*(x) - x\|^p, 0)} \leq  \theta \|c^*(x) - x\|^p.$$
Write,
\begin{multline*}
\Exp\bigl[\sum_{x\in C\setminus C'} \|\pi(c^*(x)) - \pi(x)\|^p\bigr] \leq (1+\eps)^{p}\Exp\bigl[\sum_{x\in C\setminus C'} \|c^*(x) - x\|^p\bigr] \\  +  \Exp\bigl[\sum_{x\in C\setminus C'} \max(\|\pi(c^*(x)) - \pi(x)\|^p - (1 + \eps)^{p} \|c^*(x) - x\|^p, 0)\bigr]
\end{multline*}
Now, we bound the second term as
$$
\sum_{x\in X} \E{\max(\|\pi(c^*(x)) - \pi(x)\|^p - (1 + \eps)^{p} \|c^*(x) - x\|^p, 0)}
\leq  \sum_{x\in X}\theta \|c^*(x) - x\|^p = \theta\cost_p \calC^*
$$
Therefore,
$$
\Exp\bigl[\sum_{x\in C\setminus C'} \|\pi(c^*(x)) - \pi(x)\|^p\bigr] \leq  (1+\eps)^{p}\Exp\bigl[\sum_{x\in C\setminus C'} \|c^*(x) - x\|^p\bigr] + \theta \cost_p \calC^*.
$$
We also have,
$$
\E{\sum_{x\in C \setminus C'} \|x-c^*(x)\|^p} \leq \sum_{x\in X}  \Prob{x\in C\setminus C'} \|x-c^*(x)\|^p
 \leq 2\theta\sum_{x\in X} \|x-c^*(x)\|^p = 2\theta \cost_p \calC^*.
$$
Therefore,
$$
\Exp\bigl[\ONE\set{\calE} \sum_{x\in C\setminus C'}  R_x \bigr]\leq
\Exp\bigl[\sum_{x\in C\setminus C'} \|\pi(c^*(x)) - \pi(x)\|^p + \|x-c^*(x)\|^p\bigr] \leq 5(1 + \eps)^p \theta \cost_p \calC^*.
$$
By Markov's inequality,
$$\Prob{\ONE\set{\calE} \eps^{1-p} \sum_{x\in C\setminus C'} R_x \geq c_{\alpha\epsilon\theta} \cost_p \calC^*} \leq \eps^{1-p} \cdot \frac{5(1+\eps)^p\theta
\cost_p \calC^*}{c_{\alpha\epsilon\theta} \cost_p \calC^*} = \alpha.
$$
We conclude that
$$
 \Prob{(\ref{eq:bound-main-1}) \text{ and } (\ref{eq:bound-main-2}) \text{ hold}} \geq 1 - \alpha - \Prob{\calE} \geq 1 - \alpha - \binom{k}{2}\delta.$$
\end{proof}

%% file: main-theorems.tex
\subsection{Proof of main results}
Now we prove the main results: Theorem~\ref{thm:main}, Theorem~\ref{thm:main-alt} and Theorem~\ref{main_th_informal}.

\begin{theorem}\label{thm:main}
Let $X$ be an instance of $k$-clustering with the $\ell_p$-objective.
Let $\varepsilon \in (0,1/4)$ (the distortion parameter) and $\alpha \in (0,1)$ (the failure probability).
Let $\pi:\bbR^m \to \bbR^d$ be a standard random dimension-reduction map with
\begin{equation}
d = \frac{c (\log \frac{k}{\alpha} + p\log \frac{1}{\varepsilon} + p^2)}{\varepsilon^2} \label{eq:cond-on-d}
\end{equation}
(where $c$ only depends on the parameters of map $\pi$ in Definition~\ref{def:dim-red-parameters}).

With probability at least $1-\alpha$ the following event $\cal D$ happens: for every clustering
$\calC$ of $X$ inequalities (\ref{eq:thm-1}) and (\ref{eq:thm-2}) hold.
\begin{align}
\cost_p \pi(\calC) &\leq (1+\eps)^{3p} \cost_p \calC \label{eq:thm-1}\\
(1-\varepsilon) \cost_p \calC &\leq (1+\eps)^{3p-1} \cost_p \pi(\calC) \label{eq:thm-2}
\end{align}
\end{theorem}
\begin{proof}
Let $\theta = \min(\eps^{p+1} 3^{-(p+1)(p+2)}, \alpha
\eps^p/(10k(1+\eps)^{4p-1}), 1/10^{p+1})$. We choose constant $c$ in (\ref{eq:cond-on-d}) so that the
parameters $(\varepsilon, \delta, \rho)$ of dimension reduction $\pi$ satisfy
$\delta \leq \min(\theta^7/600, \theta/k)$, $\binom{k}{2}\delta \leq \alpha
/2$, and $\rho \leq \theta$ (we can do this, since $\pi$ is a standard
dimension reduction). By our choice of parameters, $A \leq (1+\eps)^{3p-1}$
and $Ac_{\alpha\eps\theta} \leq \eps/k$ in the statement of
Theorem~\ref{thm:cost-one-cluster}.

To apply Theorem~\ref{thm:cost-one-cluster}, we define a random cluster $C$. First, assume that  event $\cal D$ does not happen. Then there exists a clustering $C_1,\dots, C_k$  violating one of the inequalities (\ref{eq:thm-1}) or (\ref{eq:thm-2}).
Thus, one of the following inequalities hold.
\begin{align}
\sum_{i=1}^k \cost_p \pi(C_i) &\geq A \left(\sum_{i=1}^k \cost_p C_i\right) + \varepsilon \cost_p \calC \\
\sum_{i=1}^k \cost_p C_i      &\geq A \left(\sum_{i=1}^k \cost_p \pi(C_i) \right) + \varepsilon \cost_p \calC
\end{align}
Therefore, for some $C_i$, we have either (\ref{eq:thm-1-one}) or (\ref{eq:thm-2-one}).
\begin{align}
\cost_p \pi(C_i) &\geq A \cost_p C_i + \frac{\varepsilon}{k} \cost_p \calC \label{eq:thm-1-one}\\
\cost_p C_i      &\geq A \cost_p \pi(C_i) + \frac{\varepsilon}{k} \cost_p \pi(\calC) \label{eq:thm-2-one}
\end{align}
Let  $C = C_i$. If $\cal D$ happens, we let $C$ be an arbitrary cluster (e.g., let $C$ consist of a single point).
Note that if $\cal D$ does not happen then we have one of the following, since $\cost_p \calC^* \leq \cost_p \calC$.
\begin{align}
\cost_p \pi(C) &\geq A \cost_p C + \frac{\varepsilon}{k} \cost_p \calC^* \label{eq:thm-1-C}\\
\cost_p C      &\geq A \cost_p \pi(C) + \frac{\varepsilon}{k} \cost_p \pi(\calC^*) \label{eq:thm-2-C}
\end{align}

Applying Theorem~\ref{thm:cost-one-cluster}, we get that the probability that one of the inequalities (\ref{eq:thm-1-C}) or (\ref{eq:thm-2-C}) holds is at most
$\alpha/2 + \binom{k}{2} \delta \leq \alpha$. Thus $\Prob{{\cal D}} \geq 1 - \alpha$.
The statement of the theorem follows.
\end{proof}

As a corollary, we get the following formulation of the theorem.

\begin{theorem}\label{thm:main-alt}
Let $X$ be an instance of $k$-clustering with the $\ell_p$-objective.
Let $\varepsilon \in (0,1/4)$ (the distortion parameter) and $\alpha \in (0,1)$ (the failure probability).
Let $\pi:\bbR^m \to \bbR^d$ be a standard random dimension-reduction map with
$$d = \frac{c_p \log (k/(\varepsilon \alpha))}{\varepsilon^2} \quad \text{where } c_p = O(p^4).$$
Then with probability at least $1-\alpha$:
$$(1-\eps) \cost_p \calC \leq \cost_p \pi(\calC) \leq (1+\eps) \cost_p \calC \quad\text{for every clustering } \calC$$
\end{theorem}
\begin{proof}
We apply Theorem~\ref{thm:main} with $\eps' = (1+\eps)^{1/(3p)} - 1 = O(\eps/p)$ and obtain the desired inequality.
\end{proof}

\medskip

Now we prove Theorem~\ref{main_th_informal}.
\begin{proof}[Proof of Theorem~\ref{main_th_informal}]
We apply Lemma~\ref{lem:sub-gauss-is-standard} to $\pi_{m,d}$. Since $\pi_{m,d}$ satisfies the condition of Theorem~\ref{distributional_jl}
and is sub-Gaussian tailed, it is a standard dimension reduction. Applying Theorem~\ref{thm:main-alt}, we get the desired result.
\end{proof}

%% file: main-measure-lemma.tex
\section{Everywhere sparse subgraph}\label{sec:everywhere-sparse}
In this section, we prove Theorem~\ref{thm:everywhere-sparse}.
The main ingredient of the proof is Lemma~\ref{lem:main-measure}, which we prove in Section~\ref{sec:combinatorial-lemma}.

\subsection{Main Combinatorial Lemma}\label{sec:combinatorial-lemma}
\begin{lemma}\label{lem:main-measure}
Let $X$ be a finite set and $V\subset X$ be a random subset of $X$. Let $\theta \in (0,1/2)$. Assume that $\Pr(x\in V)
\geq 2\theta$ for every $x$. Then there exist a random set $R\subset V$ (defined on the same probability space as $V$) and a deterministic measure $\mu$ on $X$ such
that
\begin{enumerate}
\item $\mu(x) \geq \frac{1}{|V\setminus R|}$ for every $x\in V\setminus R$ (always);
\item $\Pr(x\in R) \leq 2\theta$ for every $x\in X$;
\item $\mu(X) = \sum_{x\in X} \mu(x) \leq \frac{\Pr(V\neq \varnothing)}{\theta^2}$.
\end{enumerate}
\end{lemma}
\begin{proof}
We prove this lemma by induction on the size of the set $X$. If $|X|=0$ i.e. $X$ is empty, then properties 1--3 trivially hold. Assume
that the statement holds for all sets $X'$ of size $|X'| < |X|$ and prove it for $X$.

Let $l = \theta |X|$. Define a (deterministic) set $X'$ and a random subset $V'\subset X'$ as follows:
\begin{align}
X' &= \{x:\Pr(x\in V \text{ and } |V| < l) \geq 2\theta\} \label{def:Xprime}\\
V' &=
\begin{cases}
V\cap X', &\text{if } |V| < l\\
\varnothing, &\text{otherwise}
\end{cases}\label{def:Vprime}
\end{align}

First, we prove that for some $x_0 \in X$, $\Pr(x_0 \in V \text{ and } |V|  < l) \leq \theta$ and, consequently,
$|X'|<|X|$. To this end, we show that the average value of $\Pr(x \in V \text{ and } |V|  < l)$ for $x\in X$ is at most
$\theta$:
\begin{multline*}
\frac{1}{|X|}\sum_{x\in X} \Pr(x \in V \text{ and } |V|  < l) = \frac{1}{|X|}\sum_{x\in X} \Exp[\ONE\set{x \in V \text{ and } |V|  < l}] =\\
= \frac{1}{|X|}\Exp\Big[\sum_{x\in X} \ONE\set{x \in V \text{ and } |V|  < l}\Big]
= \frac{1}{|X|}\Exp[|V|\cdot\ONE\set{|V|  < l}]\leq \frac{l}{|X|} = \theta.
\end{multline*}
Here we used that the random variable $|V|\cdot\ONE\set{|V|  < l}$ is always less than $l$.

Since $|X'| < |X|$ and $\Pr(x\in V') = \Pr(x\in V \text{ and } |V|<l) \geq 2
\theta$ for every $x\in X'$, we
can apply the induction hypothesis to $X'$ and $V'$. By the inductive
hypothesis, there exist a random set $R' \subset X'$ and measure $\mu'$ on
$X'$ satisfying properties 1--3 for the set $X'$ and random subset $V'$.
Define a measure $\mu$ on $X$ and random subset $R\subset V$ as follows:
\begin{align}
\mu(x) &= \begin{cases}
\mu'(x) + 1/l, &\text{if } x\in X'\\
1/l, &\text{otherwise}
\end{cases}\label{def:mu}
\\
R &=
\begin{cases}
R' \cup (V\setminus X'), &\text{if } |V| < l\\
R', &\text{otherwise}
\end{cases}\label{def:R}
\end{align}

We verify that $R$ and $\mu$ satisfy the desired conditions.

\paragraph{Condition 1: for every $x\in V\setminus R$, $\mu(x) \geq 1/|V\setminus R|$ (always).} Fix $x\in V\setminus R$ and consider three cases.
If $x\in X'$ and $|V| < l$, we have $V\setminus R = V'\setminus R'$ by (\ref{def:R}). Thus, $\mu(x) > \mu'(x) \geq
\frac{1}{|V'\setminus R'|} = \frac{1}{|V\setminus R|}$ by the induction hypothesis. If $x\in X'$ and $|V| \geq l$, then
$\mu(x) \geq 1/l \geq 1/|V|$. Note that $V' = \varnothing$ (since $|V| \geq l$). Hence, $R \stackrel{\text{by
(\ref{def:R})}} = R' \subset V' = \varnothing$. In particular, $1/|V| = 1/|V\setminus R|$ and thus $\mu(x)  \geq
1/|V\setminus R|$.
\smallskip

Finally, if $x\notin X'$, then  $x\in V\setminus X' \subset R' \cup (V\setminus X')$ and $x \notin R$. Thus, $R \neq R'
\cup (V\setminus X')$. From~(\ref{def:R}), we get that $|V|\geq l$ and hence $\mu(x) = 1/l \geq 1/|V|$. Again, since
$|V| \geq l$, we have $\mu(x) \geq 1/|V| = 1/|V\setminus R|$.

\paragraph{Condition 2: $\Pr(x\in R) \leq 2\theta$.}

If $x\in X'$, then $\Pr(x\in R) = \Pr(x\in R') \leq 2\theta$ by the induction hypothesis.
If $x\notin X'$, then
$$\Pr(x\in R) \stackrel{\text{by (\ref{def:R})}}{=}  \Pr(x\in V \text{ and } |V| < l) \stackrel{\text{by (\ref{def:Xprime})}}{\leq} 2\theta .
$$

\paragraph{Condition 3: $\mu(X) = \sum_{x\in X}\mu(x) \leq \Pr(V\neq \varnothing)/\theta^2$.}
By the inductive hypothesis, we have $\mu'(X') \leq \Pr(V'\neq \varnothing)/\theta^2$. Thus,
$$\mu(X) = \mu'(X') + \frac{|X|}{l} \leq \frac{\Pr(V'\neq \varnothing)}{\theta^2} + \frac{1}{\theta}.$$
To finish the proof, we need to show that
$$\frac{\Pr(V'\neq \varnothing)}{\theta^2} + \frac{1}{\theta} \leq \frac{\Pr(V\neq \varnothing)}{\theta^2},$$
or, equivalently,  $\Pr(V\neq \varnothing) - \Pr(V'\neq \varnothing) \geq \theta$. Note that if $|V|\geq l$, then
$V\neq \varnothing$ and $V'=\varnothing$. Thus (since $V'\subset V$),
$$\Pr(V\neq \varnothing) - \Pr(V'\neq \varnothing) = \Pr(V\neq \varnothing \text{ and } V'= \varnothing)\geq \Pr(|V|\geq l).$$
Recall that for some $x_0\in X$, we have $\Pr(x_0\in V \text{ and } |V|<l)\leq \theta$. Since
for all $x\in X$, $\Pr(x\in X) \geq 2\theta$, we have
$$\Pr(|V|\geq l) \geq \Pr(x_0\in V \text{ and } |V|\geq l)\geq \Pr(x_0\in V) - \Pr(x_0\in V \text{ and } |V|< l)\geq \theta.$$
This completes the proof.
\end{proof}

\begin{corollary}\label{cor:main-measure}
Let $X$ be a finite set and $V\subset X$ be a random subset of $X$. Let $\theta \in (0,1/2)$. Then there exist a random
set $R\subset V$ and measure $\mu$ on $X$ such that
\begin{enumerate}
\item $\mu(x) \geq \frac{1}{|V\setminus R|}$ for every $x\in V\setminus R$ (always)
\item $\Pr(x\in R) \leq 2\theta$ for every $x\in X$
\item $\mu(X) = \sum_{x\in X} \mu(x) \leq \frac{1}{\theta^2}$
\end{enumerate}
\end{corollary}
\begin{proof}
Let $X' = \{x: \Pr(x\in V) \geq 2\theta\}$. We apply the lemma to $X'$
and $V'= V\cap X'$ then let
\begin{align*}
R &= R' \cup (V\setminus X'),\\
\mu(x) &=
\begin{cases}
\mu'(x), &\text{if } x\in X'\\
0, &\text{otherwise}
\end{cases}
\end{align*}
\end{proof}

We will need the following observation.
\begin{observation}\label{observ:alpha}
Let $R$ be as in Corollary~\ref{cor:main-measure} and $V_0 = V\setminus R$. Then for every $S \subset V_0$, we have
$|S| \leq \mu(S) |V_0|$.
\end{observation}
\begin{proof}
By Corollary~\ref{cor:main-measure}, $\mu(x) \geq 1/|V_0|$ for every $x\in S$. Therefore, $\mu(S) \geq |S|/|V_0|$ and
$|S| \leq \mu(S) |V_0|$.
\end{proof}

%% file: non-distorted-core.tex
\subsection{Proof of Theorem~\ref{thm:everywhere-sparse}}
\begin{proof}
Let $\beta = \frac{\theta}{1+\theta}$ and $\theta' = \theta/3$. Applying Corollary~\ref{cor:main-measure} with
parameter $\theta'$, we get a deterministic measure $\mu$ on $V$ and a random set $R \subset V$. Let $V_0 = V \setminus
R$. Denote $M = \mu(X) \leq 1/\theta'^2$. Consider the product measure $\mu^{\otimes 2}$ on $X\times X$, defined by
$\mu^{\otimes 2}((x,y)) = \mu(x) \cdot \mu(y)$ for $x,y\in X$. Then $\mu(X\times X) = M^2$.

Consider random variable $\mu^{\otimes 2}(E) \equiv \mu^{\otimes 2}(\set{(x,y):(x,y)\in E})$.
Since $\Prob{(x,y)\in E} \leq \delta$ for every $x,y\in X$, we have $\E{\mu^{\otimes 2}(E)} \leq \delta M^2$. By
Markov's inequality,
$$\Prob{\mu^{\otimes 2}(E) \geq \beta^2} \leq \frac{\delta M^2}{\beta^2} \leq \frac{5\theta}{16},$$
since $M^2 \leq 1/ \theta'^4 = 81/\theta^4$, $\beta^2 = \theta^2/(1+\theta)^2 \geq 4\theta^2/9$, and $\delta \leq \theta^7/600$ (by the condition of the theorem).

If $\mu^{\otimes 2}(E) \geq \beta^2$, we let $V' = \varnothing$. In this case, $H[V']$ is empty and trivially
$\theta$-sparse. We now consider the main case when $\mu^{\otimes 2}(E) < \beta^2$. We say that $x\in X$ is
\textit{bad} if
\begin{equation}\label{eq:defn-bad-point}
\mu\big(\set{y\in V_0: (x,y)\in E}\big) \geq \beta.
\end{equation}
Denote the set of bad points by $B$ (note that $B$ is a random set). Observe that $B\subset V$, since if $x\notin V$
then $x$ is not connected to any $y$ and thus $\set{y\in V_0: (x,y)\in E} = \varnothing$. Define $V'$ as
$$V' = V \setminus (R \cup B) = V_0 \setminus B.$$
Let us verify that $V'$ satisfies the desired properties. First, we check
that $H[V']$ is $\theta$-everywhere sparse; that is, $\deg_{H[V']} x \leq
\theta |V'|$ for every $x\in V'$ . That is, $|\set{y\in V': (x,y)\in E}|\leq
\theta |V'|$.  To this end, we upper bound the measures and cardinalities of
sets $V_0\cap B$ and $\set{y\in V': (x,y)\in E}$.
For every $x\in B$,
\begin{multline*}
\mu^{\otimes 2}\big(\set{(x,y)\in E}\big) = \mu(x) \cdot \mu\big(\set{y\in V: (x,y)\in E}\big) \geq \\
\geq \mu(x) \cdot \mu\big(\set{y\in V_0: (x,y)\in E}\big) \geq \mu(x) \cdot \beta.
\end{multline*}
Consequently,
$$\mu^{\otimes 2}(E) \geq \sum_{x \in B} \mu^{\otimes 2}\big(\set{(x,y)\in E}\big) \geq \beta \mu(B)$$
and $\mu(B) \leq \mu^{\otimes 2}(E)/\beta < \beta$. In particular,  $\mu(V_0 \cap B) \leq \mu(B) \leq \beta$.

Consider now $x\in V'$. Since $x\notin B$,
$$\mu(\set{y\in V': (x,y)\in E}) \leq \mu(\set{y\in V_0: (x,y)\in E}) \leq \beta.$$
 By Observation~\ref{observ:alpha}, $|\set{y\in V': (x,y)\in E}| \leq
\beta |V_0|$ and $|V_0 \cap B| \leq \mu(V_0\cap B) |V_0| \leq \beta |V_0|$. From the latter inequality, we get $|V'| =
|V_0\setminus B| \geq (1 - \beta) |V_0|$. We conclude that
$$|\set{y\in V': (x,y)\in E}| \leq \frac{\beta}{1-\beta} |V'| =\theta |V'|,$$
as required.

To finish the proof, we need to show that $\Prob{x\notin V\setminus V'} \leq \theta$. Note that $x$ is in $V$ but not
in $V'$ only if one of the following events happens:
\begin{itemize}
  \item $\mu^{\otimes 2}(E) \geq \beta^2$. As we showed above, the probability of this event is at most $5\theta/16$.
  \item $\mu^{\otimes 2}(E) < \beta^2$ and $x\in R$. By Corollary~\ref{cor:main-measure}, the probability of this
      event is at most $2\theta' = 2\theta/3$.
  \item $\mu^{\otimes 2}(E) < \beta^2$ and $x\in B$. By Markov's inequality, the probability of this event is at most
$$\Prob{x\in B} \leq \E{\mu\big(\set{y\in V_0: (x,y)\in E}\big)}/\beta \leq \delta M/\beta \leq \theta/300,$$
here we used that first for $x\in B$, $\mu\big(\set{y\in V_0: (x,y)\in E}\big)/\beta\geq 1$
see~(\ref{eq:defn-bad-point}), second
$$\E{\mu\big(\set{y\in V_0: (x,y)\in E}\big)}\leq \sum_{y\in X}
\Prob{(x,y)\in E} \mu(y) \leq \sum_{y\in X} \delta \mu(y) =\delta M,$$ and third $\delta \leq \theta^7/600$, $M \leq 1/ \theta'^2 = 9/\theta^2$, $\beta = \theta/(1+\theta) \geq 2\theta/3$, and $\theta \in (0, 1/2)$.
\end{itemize}
Thus, $\Prob{x\in V \setminus V'} \leq \left(\frac{5}{16} + \frac{2}{3} + \frac{1}{300}\right)\theta < \theta$.
\end{proof}

%% file: kirszbraun.tex
\section{One point extension}\label{sec:kirszbraun}
Consider two sets of points $X$ and $Y$ in Euclidean space and a one-to-one map $\varphi:X\to Y$. Suppose that for every point $x$ in $X$,
the distances from $x$ to all but a $\theta$ fraction of $x'$ in $X$ do not increase under the map $\varphi$. In this section, we show that
in this case, we have $\cost_p(Y)\leq (1+O(\theta^{1/(p+1)}))\cost_p(X)$. We prove this statement (Lemma~\ref{lem:main-Kriszbraun-corollary})
using a robust version of the classic Kirszbraun Theorem (Theorem~\ref{thm:robust-Kirszbraun}).
Later we use this Lemma to prove Theorem~\ref{cor:main-Kriszbraun-corollary}.
To state our results, we need to define
a distance expansion graph for the map $\varphi$. This definition is similar to the definition of the distortion graph.


\begin{definition}[Distance expansion graph]
Consider two finite metric spaces $(X,d_X)$ and $(Y,d_Y)$ and a map $\varphi: X \to Y$. Define the distance expansion graph for $\varphi$ on elements of the space
$X$ as follows. A pair of vertices $(x',\,x'')$ is an edge in the graph if and only if
$$d_Y(\varphi(x'),\varphi(x'')) > d_X(x',\,x'').$$
\end{definition}

\begin{theorem}[Robust one point extension theorem for $L_2$ spaces]\label{thm:robust-Kirszbraun}
Consider two finite (multi)sets of points $X\subset \bbR^{d'}$ and $Y\subset
\bbR^{d''}$ and a map $\varphi: X\to Y$. Let $G=(X,E)$ be the distance
expansion graph for $\varphi$ with respect to the Euclidean distance. Suppose
that $G$ is $\theta$-everywhere sparse. Then, for every $u\in \bbR^{d'}$ and
positive $\varepsilon$, there exists $v\in \bbR^{d''}$ such that for all but
possibly a $\theta'(\varepsilon)$ fraction of points $x$ in $X$ we have
$$\|\varphi(x) - v\| \leq  (1+\varepsilon)\|x - u\|,$$
where $\theta'(\varepsilon) = 2(1+\varepsilon)^2 \cdot \theta/\varepsilon$.
\end{theorem}

First, we show how to derive the main result of this section (Lemma~\ref{lem:main-Kriszbraun-corollary}) from  Theorem~\ref{thm:robust-Kirszbraun} and then
prove Theorem~\ref{thm:robust-Kirszbraun} itself.

\begin{lemma}\label{lem:main-Kriszbraun-corollary}
Consider two finite multisets of points $X\subset \bbR^{d'}$ and $Y\subset
\bbR^{d''}$ of the same size and a one-to-one map $\varphi: X\to Y$.
Let $G=(X,E)$ be the distance expansion graph for $\varphi$ with respect to
the Euclidean distance. Suppose that $G$ is $\theta$-everywhere sparse with
$\theta \leq 1/10^{p+1}$. Then, for every $p\geq 1$, we have the following
inequality on the cost of the clusters $X$ and $Y$:
$$\cost_p(X) \leq (1+3^{p+2}\theta^{1/(p+1)})\cost_p(Y).$$
\end{lemma}
\begin{proof}
Fix a parameter $\varepsilon = \theta^{1/(p+1)}\leq 1/10$. Let $u^*$ be the optimal center for the cluster $X$. By Theorem~\ref{thm:robust-Kirszbraun}, there exists a point $v^*\in \bbR^{d''}$ such that for
all but possibly a $\theta' = 2(1+\varepsilon)^2 \cdot \theta/\varepsilon$ fraction of points $x$ in $X$ we have
$$\|\varphi(x) - v^*\| \leq  (1+\varepsilon)\|x - u^*\|.$$
Let $\widetilde{X}$ be the set of points $x$ for which the inequality above holds. Then, $|\widetilde{X}|\geq (1-\theta')|X|$. Let us place
the center of the cluster $Y$ to $v^*$. This will give an upper bound on the cost $\cost_p(Y)$:
\begin{equation}\label{eq:main-in-Kriszbraun-corollary}
\cost_p(Y) \leq \sum_{y\in Y} \|y - v^*\|^p = \sum_{x\in X} \|\varphi(x) - v^*\|^p.
\end{equation}
We now need to estimate the right hand side of~(\ref{eq:main-in-Kriszbraun-corollary}).
For $x\in \widetilde{X}$, we already have a bound: $\|\varphi(x) - v^*\|^p \leq (1+\varepsilon)^p \|x - u^*\|^p$.
We use the following claim to bound $\|\varphi(x) - v^*\|^p$ for $x\in X\setminus \widetilde{X}$.

\begin{claim}\label{cl:Kirszbraun-rerouting}
For all $x\in X$, we have
$$\|\varphi(x) - v^*\|^p\leq (1+\varepsilon)^p \|x - u^*\|^p + \frac{3^p}{\varepsilon^{p-1}} \cdot \frac{2}{|X|}\sum_{x'\in \widetilde{X}}\|x' - u^*\|^p.$$
\end{claim}
\begin{proof}
Fix $x\in X$. Let $I_x$ be the set of its non-neighbors in the distance
expansion graph. I.e., $x'\in I_x$ if $\|\varphi(x) - \varphi(x')\| \leq \|x
- x'\|$. Since the distance expansion graph is $\theta$-everywhere sparse,
the set $I_x$ contains at least $(1-\theta)|X|$ points and  $I_x\cap
\widetilde{X}$ contains at least $(1-\theta-\theta')|X|$ points.  Consider an arbitrary $x'\in I_x\cap
\widetilde{X}$. By the triangle inequality, we have
$$\|\varphi(x) - v^*\|\leq \|\varphi(x) - \varphi(x')\| + \|\varphi(x') - v^*\|.$$
Now, $\|\varphi(x) - \varphi(x')\|\leq \|x - x'\|$ because $x'\in I_x$ and $\|\varphi(x') - v^*\|\leq (1+\varepsilon) \|x' - u^*\|$ because
$x'\in \widetilde{X}$. Thus,
$$\|\varphi(x) - v^*\| \leq \|x - x'\| + (1+\varepsilon)\|x' - u^*\|.$$
Using the triangle inequality once again, we get
$$\|\varphi(x) - v^*\| \leq \big(\|x - u^*\| + \|x' - u^*\|\big) + (1+\varepsilon)\|x' - u^*\| = \|x - u^*\| + (2+\varepsilon)\|x' - u^*\|.$$
By Lemma~\ref{ineq:sum-p-pow} applied to the inequality above,
\begin{align}
\notag \|\varphi(x) - v^*\|^p &\leq (1+\varepsilon)^p \|x - u^*\|^p + \Big(\frac{1+\varepsilon}{\varepsilon}\Big)^{p-1} (2+\varepsilon)^p \|x' - u^*\|^p\\
\label{ineq:claim-triangle-Lp}
&\leq (1+\varepsilon)^p \|x - u^*\|^p + \frac{3^p}{\varepsilon^{p-1}} \|x' - u^*\|^p,
\end{align}
here we used that $(1+\varepsilon)(2+\varepsilon) < 3$ for $\varepsilon \leq 1/10$.

We now average~(\ref{ineq:claim-triangle-Lp}) over all $x'\in I_x\cap \widetilde{X}$ and use the bound
$|I_x\cap \widetilde{X}|\geq (1-\theta-\theta') |X|\geq |X|/2$:
\begin{align*}
\|\varphi(x) - v^*\|^p &\leq (1+\varepsilon)^p \|x - u^*\|^p +
\frac{3^p}{\varepsilon^{p-1}} \frac{1}{|I_x\cap \widetilde{X}|}\sum_{x'\in I_x\cap \widetilde{X}}\|x' - u^*\|^p\\
&\leq (1+\varepsilon)^p \|x - u^*\|^p + \frac{3^p}{\varepsilon^{p-1}} \cdot \frac{2}{|X|}\sum_{x'\in \widetilde{X}}\|x' - u^*\|^p.
\end{align*}
This concludes the proof of Claim~\ref{cl:Kirszbraun-rerouting}.
\end{proof}
We now split the right hand side of (\ref{eq:main-in-Kriszbraun-corollary}) into two sums:
\begin{equation}\label{eq:main-in-Kriszbraun-corollary2}
\cost_p(Y) \leq \sum_{x\in X} \|\varphi(x) - v^*\|^p = \sum_{x\in \widetilde{X}} \|\varphi(x) - v^*\|^p + \sum_{x\in X\setminus\widetilde{X}} \|\varphi(x) - v^*\|^p.
\end{equation}
Write
$$\sum_{x\in \widetilde{X}} \|\varphi(x) - v^*\|^p\leq (1+\varepsilon)^p\sum_{x\in \widetilde{X}} \|x - u^*\|^p,$$
and, using Claim~\ref{cl:Kirszbraun-rerouting},
\begin{eqnarray*}
\sum_{x\in X\setminus \widetilde{X}} \|\varphi(x) - v^*\|^p &\leq& (1+\varepsilon)^p\sum_{x\in X\setminus \widetilde{X}} \|x - u^*\|^p
+ \frac{3^p}{\varepsilon^{p-1}} \cdot \frac{2}{|X|}\sum_{\substack{x\in X\setminus\widetilde{X} \\ x'\in \widetilde{X}}}\|x' - u^*\|^p\\
&=& (1+\varepsilon)^p\sum_{x\in X\setminus \widetilde{X}} \|x - u^*\|^p
+ \frac{3^p}{\varepsilon^{p-1}} \cdot \frac{2|X\setminus \widetilde{X}|}{|X|}\sum_{x'\in \widetilde{X}}\|x' - u^*\|^p\\
&\leq&(1+\varepsilon)^p\sum_{x\in X\setminus \widetilde{X}} \|x - u^*\|^p + \frac{3^p}{\varepsilon^{p-1}} \cdot 2\theta'\sum_{x'\in \widetilde{X}}\|x' - u^*\|^p.
\end{eqnarray*}
Here we used that $|X\setminus \widetilde{X}|\leq \theta' |X|$. We now have bounds for the both terms in the right
hand side of inequality~(\ref{eq:main-in-Kriszbraun-corollary2}). We plug them in and obtain the following upper bound on $\cost_p(Y)$:
\begin{align*}
\cost_p(Y) &\leq
(1+\varepsilon)^p\sum_{x\in X} \|x - u^*\|^p + \frac{3^p}{\varepsilon^{p-1}} \cdot 2\theta'\sum_{x'\in \widetilde{X}}\|x' - u^*\|^p\\
&\leq \left((1+\varepsilon)^p + \frac{2\theta' \cdot 3^p}{\varepsilon^{p-1}}\right) \cost_p(X).
\end{align*}
Observe that for $\varepsilon \leq 1/10$, one has $(1+\varepsilon)^p \leq 1 + 3^p \varepsilon$ and
$$\frac{2\theta' \cdot 3^p}{\varepsilon^{p-1}} = \frac{4\theta (1+\varepsilon)^2 \cdot 3^p}{\varepsilon^{p}}
= 4 (1+\varepsilon)^2 \cdot 3^p \varepsilon \leq 7\cdot 3^p \varepsilon.$$
Therefore,
$$\cost_p(Y) \leq  (1 + 3^{p+2} \varepsilon) \,\cost_p(X).$$
\end{proof}

\subsection{Proof of Theorem~\ref{thm:robust-Kirszbraun}}
We prove Theorem~\ref{thm:robust-Kirszbraun} using a duality argument. Fix a positive $\varepsilon$; and denote the size of $X$ by $n$. Let $\eta = (\theta'(\varepsilon) n)^{-1}$.
Consider the following convex polytope:
$$\Lambda_{\eta} = \{\lambda\in \bbR^X: \sum_{x\in X} \lambda_x = 1; 0 \leq \lambda_{x'} \leq \eta \text{ for all } x'\in X\}.$$
For every $\lambda \in \Lambda_{\eta}$, $u'\in \bbR^{d'}$, and $v'\in \bbR^{d''}$, let
\begin{equation}\label{eq:defn-f}
f(X, \lambda, u') = \sum_{x\in X} \lambda_x \|u' - x\|^2\;\;\; \text{and}\;\;\; f(\varphi(X), \lambda, v') = \sum_{x\in X} \lambda_x \|v' - \varphi(x)\|^2.
\end{equation}
That is, $f(X, \lambda, u')$ is the cost of a single cluster $X$ with a center in $u'$ according to the weighted $k$-means objective. The weight of a point
$x\in X$ is $\lambda_x$. Similarly, $f(\varphi(X), \lambda, v')$ is the cost of the cluster $\varphi(X)$ with a center in $v'$.
The optimal centers for clusters $X$ and $\varphi(X)$ are located at the centers of mass of $X$ and $\varphi(X)$ respectively. Thus, for
a given $\lambda\in \Lambda_{\eta}$, $X$, and $\varphi$, the objective functions $f(X, \lambda, u')$ and $f(\varphi(X), \lambda, v')$ are minimized when
$u' = \sum_x \lambda_x x$ and $v' = \sum_x \lambda_x \varphi(x)$. Consequently (see Section~\ref{closed-form-expr-for-kmeans-cluster} for details),
\begin{align}
\label{eq:opt-weighted-kmeans1} \min_{u'\in \bbR^{d'}} f(X, \lambda, u') &= \sum_{(x',\,x'')\in P} \lambda_{x'}\lambda_{x''} \|x' - x''\|^2;\; \text{and}\\
\label{eq:opt-weighted-kmeans2} \min_{v'\in \bbR^{d''}} f(\varphi(X), \lambda, v') &= \sum_{(x',\,x'')\in P} \lambda_{x'}\lambda_{x''} \|\varphi(x') - \varphi(x'')\|^2,
\end{align}
where $P$ is the set of all \emph{unordered} pairs $(x',\,x'')$ with $x',x''\in X$.

\medskip

Let $F(v', \lambda) = (1+\varepsilon)^2 f(X, \lambda, u) - f(\varphi(X), \lambda, v')$. Our goal is to show that
\begin{equation}\label{eq:minmax-main}
\max_{v'\in \bbR^{d''}} \min_{\lambda\in \Lambda_{\eta}} F(v', \lambda) \geq 0.
\end{equation}

\begin{lemma}\label{lem:minmax-kirszbraun}
Inequality~(\ref{eq:minmax-main}) holds.
\end{lemma}
\begin{proof}
Observe that functions $f(X,\lambda, u)$ and $f(\varphi(X),\lambda, v)$ are linear in $\lambda$ for fixed $u$ and $v$ and convex in
$u$ and $v$ respectively for a fixed $\lambda$ (see~(\ref{eq:defn-f})). Hence, $v' \mapsto F(v', \lambda)$ is a concave function for every $\lambda$;
and $\lambda \mapsto F(v', \lambda)$ is a linear function for every $v'$. Therefore, by Sion's minimax theorem (a variant of von Neumann's minimax theorem)~\cite[Theorem 4.2]{Sion58}, we have
$$\max_{v'\in \bbR^{d''}} \min_{\lambda \in \Lambda_{\eta}} F(v', \lambda) = \min_{\lambda\in \Lambda_{\eta}} \max_{v'\in \bbR^{d''}} F(v', \lambda).$$
Thus, it suffices to prove that $\max_{v'\in \bbR^{d''}} F(v', \lambda)\geq 0$ for all $\lambda\in \Lambda_{\eta}$.
Fix $\lambda\in\Lambda_{\eta}$. Then,
\begin{eqnarray*}
\max_{v'\in R^{d''}} F(v', \lambda) &=& \max_{v'\in R^{d''}} (1+\varepsilon)^2 f(X, \lambda, u) - f(\varphi(X), \lambda, v')\\
&=& (1+\varepsilon)^2 f(X, \lambda, u) - \min_{v'\in R^{d''}} f(\varphi(X), \lambda, v')\\
&\geq& \min_{u'\in R^{d''}} (1+\varepsilon)^2 f(X, \lambda, u') - \min_{v'\in R^{d''}} f(\varphi(X), \lambda, v').
\end{eqnarray*}
Using formulae (\ref{eq:opt-weighted-kmeans1}) and (\ref{eq:opt-weighted-kmeans2}) for the minimum of the function $f$, we have
\begin{eqnarray*}
\max_{v'\in R^{d''}} F(v', \lambda) &\geq&
(1+\varepsilon)^2 \sum_{(x',\,x'')\in P} \lambda_{x'}\lambda_{x''} \|x' - x''\|^2\ - \sum_{(x',\,x'')\in P} \lambda_{x'}\lambda_{x''} \|\varphi(x') - \varphi(x'')\|^2\\
&=& \sum_{(x',\,x'')\in P} \lambda_{x'}\lambda_{x''}\Big[(1+\varepsilon)^2 \|x' - x''\|^2\ - \|\varphi(x') - \varphi(x'')\|^2\Big].
\end{eqnarray*}
We now split the sum on the right hand side into two parts: the sum over pairs $(x',\,x'')\in E$ and pairs $(x',\,x'')\notin E$. Then,
we upper bound each term in each of the sums. For $(x',\,x'')\in E$, we use a trivial bound
$$(1+\varepsilon)^2 \|x' - x''\|^2\ - \|\varphi(x') - \varphi(x'')\|^2\geq - \|\varphi(x') - \varphi(x'')\|^2.$$
For $(x',\,x'')\notin E$, we have $\|x' - x''\|^2\ \geq \|\varphi(x') - \varphi(x'')\|^2$ by the definition of the distance expansion graph, and
hence
$$(1+\varepsilon)^2\|x' - x''\|^2\ - \|\varphi(x') - \varphi(x'')\|^2\geq ((1+\varepsilon)^2 - 1) \|\varphi(x') - \varphi(x'')\|^2.$$
Denote $\varepsilon' = (1+\varepsilon)^2 - 1$. We obtain the following bound:
\begin{eqnarray}\label{eq:eps-1+eps}
\max_{v'\in R^{d''}} F(v', \lambda) &\geq& \varepsilon' \sum_{(x',\,x'')\in P\setminus E}  \lambda_{x'}\lambda_{x''} \|\varphi(x') - \varphi(x'')\|^2 -
\sum_{(x',\,x'')\in E}  \lambda_{x'}\lambda_{x''} \|\varphi(x') - \varphi(x'')\|^2\\
&=&
\notag \varepsilon' \sum_{(x',\,x'')\in P}  \lambda_{x'}\lambda_{x''} \|\varphi(x') - \varphi(x'')\|^2 -
(1+\varepsilon')\sum_{(x',\,x'')\in E}  \lambda_{x'}\lambda_{x''} \|\varphi(x') - \varphi(x'')\|^2.
\end{eqnarray}
We estimate the second sum  using the following claim.
\begin{claim}\label{lem:conv-comb-traingle-ineq}
For all $x', x''\in X$ and $\lambda \in \Lambda_{\eta}$, we have
$$\|\varphi(x') - \varphi(x'')\|^2\leq 2\sum_{x\in X} \lambda_x \|\varphi(x) - \varphi(x')\|^2  + \lambda_x\|\varphi(x) - \varphi(x'')\|^2.$$
\end{claim}
\begin{proof}
The desired inequality is the convex combination of the relaxed triangle inequalities for squared Euclidean distances (see Corollary~\ref{cor:relax-triangle-inequality}):
$$\|\varphi(x') - \varphi(x'')\|^2\leq 2\Big[\|\varphi(x) - \varphi(x')\|^2  + \|\varphi(x) - \varphi(x'')\|^2\Big]$$
with weights $\lambda_x$. Note that $\sum_{x\in X} \lambda_x = 1$ for each $\lambda\in \Lambda_{\eta}$.
\end{proof}

\medskip

\noindent By Claim~\ref{lem:conv-comb-traingle-ineq},
\begin{align*}
\sum_{(x',\,x'')\in E}  \lambda_{x'}\lambda_{x''} \|\varphi(x') - \varphi(x'')\|^2 &\leq
2\sum_{(x',\,x'')\in E}\lambda_{x'}\lambda_{x''} \sum_{x\in X} \Big(\lambda_x \|\varphi(x) - \varphi(x')\|^2  + \lambda_x\|\varphi(x) - \varphi(x'')\|^2\Big)\\
&=2\sum_{x, x'\in X} \left[\sum_{x'': (x',\,x'')\in E}\lambda_{x''}\right] \lambda_x \lambda_{x'} \|\varphi(x) - \varphi(x')\|^2.
\end{align*}
Since the degree of every vertex $x'$ in the distance expansion graph is at most $\theta n$ and each $\lambda_{x''}\leq \eta$, we have
$\left[\sum_{x'': (x',\,x'')\in E}\lambda_{x''}\right]\leq \theta\eta n$. Therefore,
\begin{multline*}
\sum_{(x',\,x'')\in E}  \lambda_{x'}\lambda_{x''} \|\varphi(x') - \varphi(x'')\|^2 \leq
2\theta \eta n \sum_{x,x'\in X} \lambda_x \lambda_{x'} \|\varphi(x) - \varphi(x')\|^2 = \\=
4\theta \eta n \sum_{(x',x'')\in P} \lambda_{x'} \lambda_{x''} \|\varphi(x') - \varphi(x'')\|^2.
\end{multline*}
Note that $4\theta \eta n = 2\varepsilon/(1+\varepsilon)^2 < \varepsilon'/(1+\varepsilon')$.
We plug in the bound above in Equation~(\ref{eq:eps-1+eps}) and obtain the desired inequality:
$$
\max_{v'\in R^{d''}} F(v', \lambda) \geq
\varepsilon' \sum_{(x',\,x'')\in P}  \lambda_{x'}\lambda_{x''} \|\varphi(x') - \varphi(x'')\|^2 -
\frac{(1+\varepsilon')\cdot\varepsilon'}{1+\varepsilon'} \sum_{(x',\,x'')\in P}  \lambda_{x'}\lambda_{x''} \|\varphi(x') - \varphi(x'')\|^2 = 0.
$$
\end{proof}
Let us now return to the proof of Theorem~\ref{thm:robust-Kirszbraun}.
Let $v$ be the point that maximizes the functional $\min_{\lambda\in \Lambda_{\eta}} F(v, \lambda)$. By Lemma~\ref{lem:minmax-kirszbraun},
$\min_{\lambda\in \Lambda_{\eta}} F(v, \lambda) \geq 0$ or, in other words, $F(v, \lambda) \geq 0$ for all $\lambda \in \Lambda_{\eta}$.
Consider the set $S=\{x\in X: \|\varphi(x)-v \| > (1+\varepsilon) \|x - u\|\}$. If $S=\varnothing$ then we are done. Otherwise,
define a vector $\lambda^*$ as follows
$$\lambda^*_x =\begin{cases}
                1/|S|, & \mbox{if } x\in S\\
                0, & \mbox{otherwise}.
              \end{cases}$$
Note that $(1+\varepsilon)^2 \|x - u\|^2 - \|\varphi(x)-v \|^2$ is negative for all $x\in S$ (by the definition of $S$). Hence,
$$F(v, \lambda^*) = \frac{1}{|S|} \sum_{x\in S} (1+\varepsilon)^2 \|x - u\|^2 - \|\varphi(x)-v \|^2 < 0$$
and, consequently, $\lambda^*\notin \Lambda_{\eta}$. Therefore, $1/|S|>\eta$ (otherwise, $\lambda^*$ would belong to $\Lambda_{\eta}$) and
$|S| < 1/\eta = \theta'(\varepsilon) n$.
This finishes the proof of Theorem~\ref{thm:robust-Kirszbraun} since for all $x\notin S$, we have $\|\varphi(x)-v \| \leq (1+\varepsilon) \|x - u\|$.

\begin{proof}[Proof of Theorem~\ref{cor:main-Kriszbraun-corollary}]
Let $X = \tilde C$ and $Y =\varphi(\tilde C)$. We apply Lemma~\ref{lem:main-Kriszbraun-corollary} to maps $(1+\varepsilon) \varphi$ and $(1+\varepsilon)\varphi^{-1}$ and get the desired result.
\end{proof}

%% file: appendix.tex
\appendix
\section{Inequality for the sum of \texorpdfstring{$p$}{p}-th powers}
\label{sec:ineq:sum-p-pow}
\begin{lemma}\label{ineq:sum-p-pow}
Let $x$ and $y_1, \dots, y_r$ be non-negative real numbers, and $\varepsilon > 0$, $p \geq 1$.
Then
$$
\left(x+\sum_{i=1}^r y_i\right)^p \leq (1+\eps)^{p-1} x^p +  \left(\frac{(1+\eps)r}{\eps}\right)^{p-1} \sum_{i=1}^r y_i^p.$$
\end{lemma}
\begin{proof}
Let $t = \frac{1}{1+\eps}$. Write,
$$\left(x + \sum_{i=1}^r y_i\right)^p = \frac{1}{t^p} \left(t x + \sum_{i=1}^r \frac{1-t}{r} \left(\frac{r\, t \,  y_i}{1-t}\right)\right)^p.
$$
The expression in the parentheses on the right is a convex combination of numbers $x, \frac{r t y_1}{1-t} ,\dots, \frac{r t y_r}{1-t}$ with coefficients $t, \frac{1-t}{r}, \dots, \frac{1-t}{r}$ (which add up to 1). Applying Jensen's inequality, we get
$$\frac{1}{t^p} \left(t x + \sum_{i=1}^r \frac{1-t}{r} \left(\frac{r\, t\,  y_i}{1-t}\right)\right)^p
\leq \frac{t x^p}{t^p} +\frac{1-t}{r\,t^p} \sum_{i=1}^r \left(\frac{r\, t\, y_i}{1-t}\right)^p =
(1+\eps)^{p-1} x^p +  \left(\frac{(1+\eps)r}{\eps}\right)^{p-1} \sum_{i=1}^r y_i^p.$$
\end{proof}

\begin{corollary}[Relaxed triangle inequalities]\label{cor:relax-triangle-inequality}
For any vectors $u$, $v$, $w$ and numbers $\varepsilon > 0$, $p\geq1$, we have
\begin{enumerate}
\item
$$\|u-w\|^p \leq (1+\varepsilon)\|u-v\|^p + \left(\frac{(1+\eps)}{\eps}\right)^{p-1} \|v-w\|^p.$$
\item
$$\|u-w\|^2 \leq 2 \|u-v\|^2 + 2\|v-w\|^2.$$
\end{enumerate}
\end{corollary}
\begin{proof}
Using Lemma~\ref{ineq:sum-p-pow}, we get
$$\|u-w\|^p \leq (\|u-v\| + 2\|v-w\|)^p \leq (1+\varepsilon)\|u-v\|^p + \left(\frac{(1+\eps)}{\eps}\right)^{p-1} \|v-w\|^p.$$
Item 2 is a special case of this inequality with $\varepsilon = 1$ and $p=2$.
\end{proof}

\section{Closed-form expression for the cost of a cluster}\label{closed-form-expr-for-kmeans-cluster}
In this section, we derive a well-known
formula~(\ref{eq:opt-weighted-kmeans1}) for computing the cost of a cluster
with respect to the $k$-means objective. The optimal center of the cluster
formed by points in set $X$ with weights $\lambda_x$ is located in the center
of mass of $X$. Let $a$ and $b$ be i.i.d random variables such as $\Pr(a =
x)=\Pr(b=x) = \lambda_x$. Then, the cost of the cluster $X$ equals $\Exp[(a -
\Exp a)^2]= \Var[a]$ and
$$\sum_{(x',\,x'')\in P} \lambda_{x'}\lambda_{x''} \|x' - x''\|^2 = \frac{1}{2}\sum_{(x',\,x'')\in X\times X} \lambda_{x'}\lambda_{x''} \|x' - x''\|^2
= \frac{1}{2}\Exp\|a - b\|^2 = \Exp a^2 - (\Exp a)^2 = \Var[a].$$

%% file: tail-bound.tex
\section{Sub-Gaussian Tailed Dimension Reduction}
In this section, first we prove that every sub-Gaussian tailed dimension reduction is standard. Then we show that
the Gaussian dimension reduction is sub-Gaussian tailed.
\begin{lemma}\label{lem:sub-gauss-is-standard}
Let $\varepsilon < 1/2$.
Assume that a family of random maps $\pi_{m,d} \colon \bbR^m \to \bbR^d$ satisfies the condition of Theorem~\ref{distributional_jl}
and is sub-Gaussian tailed (satisfies Definition~\ref{def_standard}).
Then $\pi_{m,d}$ is a standard dimension reduction (see Definition~\ref{def:dim-red-parameters}).
\end{lemma}
\begin{proof}
Denote the parameters of dimension reduction $\pi_{m,d}$ by $(\varepsilon, \delta, \rho)$.
Since $\pi_{m,d}$ satisfies the condition of Theorem~\ref{distributional_jl},
$\delta \leq \exp(-c \varepsilon^2 d)$ for some $c$.

Let $u$ be a unit vector in $\bbR^m$ and $\xi = \|\pi(u)\| - 1$. Since
$\pi_{m,d}$ is sub-Gaussian tailed,  $\Prob{\xi > t} \leq \exp(-c t^2 d)$ for
some $c$. We assume that $d \geq c(p-1)/\eps^2$. We have,
\begin{align*}
\Exp&\left[\ONE\set{\set{\|\pi(u)\| > (1+\eps) \|u\|} }
\Bigl(\frac{\|\pi(u)\|^p}{\|u\|^p} - (1+\eps)^p \Bigr) \right]\\
 &=
\E{ \ONE\set{\xi > {\eps}} ((1+\xi)^p - (1+\eps)^p)}
=
\int_{\eps}^{\infty} ((1+t)^p - (1+\eps)^p)\, d\Prob{\xi \leq t}\\
&\stackrel{\text{\tiny integration by parts}}{=}
\left.\Bigl(((1+t)^p - (1+\eps)^p)\cdot (-\Prob{\xi > t})\Bigr)\right|^{t=\infty}_{t=\varepsilon} +
\int_{\eps}^{\infty} p(1+t)^{p-1} \Prob{\xi > t} dt\\
&\leq 0 + \int_{\eps}^{\infty} p(1+t)^{p-1} e^{-c t^2 d} dt = \int_{\eps}^{\infty} (p(1+t)^{p-1} e^{-c t^2 d/2}) e^{-c t^2 d/2}dt
\end{align*}
By differentiating $g(t) = p(1+t)^{p-1} e^{-c t^2 d/2}$ by $t$, we get that $g'(t)$  is decreasing when
$t(1+t) \geq \frac{p-1}{cd}$. Since $d \geq \frac{c(p-1)}{\eps^2}$,
$g(t)$ attains its maximum on $[\eps, \infty)$
when $t =\eps$. We have
$$g(t) \leq p(1+\eps)^{p-1} e^{-c \eps^2 d/2} \leq
e^{-(p-1)/\eps^2 + \ln p + (p-1)\ln (1+\eps)} \leq 1$$
(here, we used that $\eps \leq 1/2$ and $d \geq c(p-1)/\eps^2$.
Therefore,
\begin{align*}
\E{ \ONE\set{\xi > {\eps}} ((1+\xi)^p - (1+\eps)^p)} &\leq
\int_{\eps}^{\infty} e^{-c t^2 d/2}dt \leq
\int_{\eps}^{\infty} \frac{t}{\eps}\, e^{-c t^2 d/2}dt \\ &=
\frac{1}{c d \eps} \int_{\eps^2/2}^{\infty}  e^{-c d s}ds =
\frac{1}{c d \eps} e^{-c d \eps^2/2} < e^{-c d \eps^2/2}.
\end{align*}
\end{proof}

Consider a $d\times m$ matrix $G$ whose entries are i.i.d. Gaussian random variables ${\cal N}(0,1)$. Matrix $G$ defines a linear dimension reduction $\pi(u) = \frac{Gu}{\sqrt{d}}$~\cite{indyk1998approximate}.
\begin{claim}\label{claim:gaussian-is-standard}
The Gaussian dimension reduction $\pi$ (defined above) is sub-Gaussian tailed.
\end{claim}
\begin{proof}
Let $u$ be a unit vector in $\bbR^m$, $\xi = \|\pi(u)\| - 1$ and $\eta = \sqrt{d}(\xi+1) = \sqrt{d}\|\pi(u)\|$.
Note that $\eta^2$ has the $\chi^2$-distribution with $d$ degrees of freedom.
As was shown by  Laurent and Massart~\cite[Lemma 1, Inequality (4.3)]{LM00},
$\Prob{\eta^2 - d \geq 2\sqrt{d\cdot x} + 2x} \leq e^{-x}$ for any positive $x$.  Plugging in $x =  t^2 d/2$ (where $t > 0$),
we get
\begin{align*}
\Prob{\xi \geq t} &= \Prob{(\xi + 1)^2  \geq 1+ 2t + t^2} \leq
\Prob{(\xi + 1)^2 - 1 \geq \sqrt{2}t + t^2}\\&
 = \Prob{\eta^2 - d \geq 2\sqrt{dx} + 2x} \leq e^{-x} = e^{-t^2d/2}.
\end{align*}
\end{proof}

%% file: fast-jl.tex
\section{Fast dimension reduction}
\label{fast_jl_appendix}

In this section, we show that a version of Theorem~\ref{main_th_informal} holds
for \emph{fast} dimension reduction maps introduced by Ailon and
Chazelle~\cite{ailon2006approximate} albeit
with somewhat worse parameters. More precisely, the main result of this
section is the following theorem.

\begin{theorem}
  \label{fast_jl_main_thm}
  There exists a family of maps $\pi_{m, d} \colon \mathbb{R}^m \to \mathbb{R}^d$
  that take time $O(m \log m)$ to apply to a given vector $x \in \mathbb{R}^m$
  such that the conclusion of Theorem~\ref{main_th_informal} holds with
  \begin{equation}
  \label{fast_jl_bound}
  d = O\left(\frac{p^6 \cdot \log^2 \frac{k}{\varepsilon \delta}}{\varepsilon^2}\right).
  \end{equation}
\end{theorem}

Let us first describe the construction of the map $\pi_{m,d}$.
By padding vectors with zeros, we can assume that $m$ is a power of two.
Let $D$ be a diagonal $m \times m$ matrix with i.i.d.\ uniform $\pm 1$ entries,
$H$ be a normalized Hadamard transform (that is, $H$ is an orthogonal matrix
with all the entries being $\pm 1 / \sqrt{m}$, it can be chosen
in a way that the map $z \mapsto Hz$ can be computed in time $O(m \log m)$~\cite{ailon2006approximate}), and $S$
is a diagonal matrix with i.i.d.\ entries such that
$\mathrm{Pr}\left[S_{ii} = \sqrt{m / d}\right] = d / m$ and $0$ with the remaining probability.
Denote $\Pi = SHD$, which is a map from $\mathbb{R}^m$ to $\mathbb{R}^m$.
For every unit vector $x$, we have $\mathbb{E}[\|\Pi x\|_2^2] = 1$.
The image of $\Pi$ is a random subspace spanned by several standard basis vectors; note that its
expected dimension is $d$. We define $\pi_{m, d}$ as the projection on this subspace.
Note that this map is not exactly what we want, since the dimension of the image
of $\pi_{m, d}$ is $d$ merely \emph{in expectation.}
However, with probability at least $0.9$, the dimension is at most $10d$,
and conditioning on this event can increase the probability of any event
by at most $10/9$, thus it is enough to show that
Theorem~\ref{fast_jl_main_thm} holds for $\Pi$ if we set $d$ as in~(\ref{fast_jl_bound}).

We start the proof with a statement that is implicitly contained
in~\cite{cohen2015optimal}.
Let us fix  a $p \geq 1$ and denote $\|Y\|_{L_p} = \mathbb{E}\left[|Y|^p\right]^{1/p}$ for a random variable $Y$.

\begin{theorem}[Implicit in Cohen, Nelson, and Woodruff~\cite{cohen2015optimal}]
  \label{moment_bound}
  Let $Y$ be a random variable.
  Let $X_1, X_2, \ldots, X_n$ be non-negative random variables that are also independent
  conditioned on $Y$.
  Denote $S = \mathbb{E}\Bigl[\sum_i X_i\Bigr]$ and $K = p \cdot \left\|\max_i X_i\right\|_{L_p}$.
  Suppose that $\mathbb{E}\Bigl[\sum_i X_i \mathrel{\Big|} Y\Bigr] = S$ holds almost surely.
  Then,
  $$
  \Bigl\|\sum_i X_i - S\Bigr\|_{L_p} \lesssim K + \sqrt{KS}.
  $$
\end{theorem}
\begin{proof}
  First, we show an auxiliary claim.
  \begin{claim}\label{claim:cor-khintchine}
    We have,
    $$
    \Bigl\|\sum_i X_i - S\Bigr\|_{L_p} \lesssim \sqrt{p} \cdot \biggl\|\Bigl(\sum_i X_i^2\Bigr)^{1/2}\biggr\|_{L_p}.
    $$
  \end{claim}
  \begin{proof}
    Write, using that $S = \mathbb{E}\Bigl[\sum_i X_i \mathrel{\Big|} Y \Bigr]$ almost surely,
    $$
      \Bigl\|\sum_i X_i - S\Bigr\|_{L_p} = \mathbb{E}\biggl[\Bigl|\sum_i X_i - \mathbb{E}\Bigl[\sum_i X_i \mathrel{\Big|} Y\Bigr]\Bigr|^p\biggr]^{1/p}
       = \mathbb{E} \biggl[\mathbb{E} \biggl[\Bigl|\sum_i \bigl(X_i - \mathbb{E}[X_i \mathrel{|} Y]\bigr)\Bigr|^p \mathrel{\bigg|} Y \biggr]\biggr]^{1/p}.
    $$
    Let $\sigma_1, \sigma_2, \ldots, \sigma_n$ be independent uniform $\{\pm 1\}$-valued random variables that are independent
    from $X_i$'s and $Y$. Then, by the symmetrization lemma for \emph{independent} random variables (see e.g., Lemma~6.3 in~\cite{ledoux2013probability}
    or Lemma 6.4.2 in \cite{Versh18})
    applied to random variables $X_i$ conditioned on $Y$, we have
    $$\mathbb{E} \biggl[\Bigl|\sum_i \bigl(X_i - \mathbb{E}[X_i \mathrel{|} Y]\bigr)\Bigr|^p \mathrel{\bigg|} Y \biggr]^{1/p}
    \leq 2\, \mathbb{E}\biggl[\Bigl|\sum_i \sigma_i X_i\Bigr|^p \mathrel{\bigg|} Y\biggr]^{1/p}.$$
    Thus,
    \begin{align*}
     \Bigl\|\sum_i X_i - S\Bigr\|_{L_p} &\leq 2 \mathbb{E}\biggl[\mathbb{E}\biggl[\Bigl|\sum_i \sigma_i X_i\Bigr|^p \mathrel{\bigg|} Y\biggr]\biggr]^{1/p}
       = 2\,  \mathbb{E}\biggl[\Bigl|\sum_i \sigma_i X_i\Bigr|^p\biggr]^{1/p}\\
      & = 2\, \mathbb{E}\biggl[\mathbb{E}\biggl[\Bigl|\sum_i \sigma_i X_i\Bigr|^p \mathrel{\bigg|} X_1, \ldots, X_n \biggr]\biggr]^{1/p}.
    \end{align*}
    By the Khintchine inequality (see e.g., Lemma~4.1 in~\cite{ledoux2013probability}), we have
    $$\mathbb{E}\biggl[\Bigl|\sum_i \sigma_i X_i\Bigr|^p \mathrel{\bigg|} X_1, \ldots, X_n \biggr]\lesssim p^{p/2} \cdot \Bigl(\sum_i X_i^2\Bigr)^{p/2}.$$
    Thus,
    $$\Bigl\|\sum_i X_i - S\Bigr\|_{L_p} \lesssim \sqrt{p} \cdot \mathbb{E}\biggl[\Bigl(\sum_i X_i^2\Bigr)^{p/2}\biggr]^{1/p}
      = \sqrt{p} \cdot \biggl\|\Bigl(\sum_i X_i^2\Bigr)^{1/2}\biggr\|_{L_p}.$$
  \end{proof}

  We now return to the proof of Theorem~\ref{moment_bound}. By Claim~\ref{claim:cor-khintchine} and the Cauchy--Schwartz inequality, we have
  \begin{align*}
  \Bigl\|\sum_i X_i - S\Bigr\|_{L_p} &\lesssim \sqrt{p} \cdot \biggl\|\Bigl(\sum_i X_i^2\Bigr)^{1/2}\biggr\|_{L_p}
  \leq \sqrt{p} \cdot \biggl\|\bigl(\max_i X_i\bigr)^{1/2} \cdot \Bigl(\sum_i X_i\Bigr)^{1/2}\biggr\|_{L_p}\\
  &\leq \underbrace{\sqrt{p} \cdot \Bigl\|\max_i X_i\Bigr\|_{L_p}^{1/2}}_{\sqrt{K}} \cdot \Bigl\|\sum_i X_i\Bigr\|_{L_p}^{1/2}
  =\sqrt{K}\cdot \Bigl\|\sum_i X_i\Bigr\|_{L_p}^{1/2}.
  \end{align*}
  Observe that $d(X,Y) = \|X-Y\|_{L_p}$ is a metric, and, consequently, $\sqrt{d(X,Y)} = \|X-Y\|_{L_p}^{1/2}$
  is also a metric. By the triangle inequality for $\|\cdot\|_{L_p}^{1/2}$, we have
  $$
    \sqrt{K} \cdot \Bigl\|\sum_i X_i\Bigr\|_{L_p}^{1/2}
    \leq \sqrt{K} \cdot \biggl(\|S\|_{L_p}^{1/2} + \Bigl\|\sum_i X_i - S\Bigr\|_{L_p}^{1/2}\biggr)
    \leq \sqrt{K} \cdot \biggl(\sqrt{S} + \Bigl\|\sum_i X_i - S\Bigr\|_{L_p}^{1/2}\biggr).
  $$
Therefore, we obtained the following inequality:
$$\Bigl\|\sum_i X_i - S\Bigr\|_{L_p} \lesssim \sqrt{K} \cdot \biggl(\sqrt{S} + \Bigl\|\sum_i X_i - S\Bigr\|_{L_p}^{1/2}\biggr).$$
It implies that
$\Bigl\|\sum_i X_i - S\Bigr\|_{L_p} \lesssim \sqrt{KS}$ or
$\Bigl\|\sum_i X_i - S\Bigr\|_{L_p} \lesssim \sqrt{K} \cdot \Bigl\|\sum_i X_i - S\Bigr\|_{L_p}^{1/2}$.
The desired inequality follows.
\end{proof}

The next theorem bounds higher moments of $\bigl\|\Pi x\bigr\|_2^2$. It is also
implicitly proved by~Cohen, Nelson, and Woodruff~\cite{cohen2015optimal}.

\begin{theorem}[Implicit in Cohen, Nelson, and Woodruff~\cite{cohen2015optimal}]
  \label{fast_jl_moment_bound}
  Fix $p \geq 1$ and $x \in \mathbb{R}^m$ with $\|x\|_2 = 1$. Denote $K = \frac{p^2 + p \log d}{d}$.
  Then,
  $$
  \Bigl\|\bigl\|\Pi x\bigr\|_2^2 - 1\Bigr\|_{L_p} \lesssim K + \sqrt{K}.
  $$
\end{theorem}
\begin{proof}
  Denote $z= HDx$ and $\eta_i = \ONE\{S_{ii} \ne 0\}$.
  Then, we have $\|z\|_2 = 1$ and $\|\Pi x\|_2^2 = \frac{m}{d} \cdot \sum_i \eta_i z_i^2$.
  We apply Theorem~\ref{moment_bound} to $X_i = \frac{m}{d} \cdot \eta_i z_i^2$,
  $Y = z$, and $S = 1$ (all the conditions are straightforward to check).
  Now, to obtain the required bound, it suffices to show that
  $$
  p \cdot \|\max_i X_i\|_{L_p} \equiv \frac{pm}{d} \cdot \bigl\|\max_i \eta_i z_i^2\bigr\|_{L_p} \lesssim \frac{p^2 + p \log d}{d}.
  $$
  Denote $q = \max\{p, \log d\}$. Using the Khintchine inequality (recall that $D_{jj}$ is a uniform $\pm1$-valued random variable), we get that for every $i$
\begin{align}\label{ineq:z-i}
  \E{z_i^{2q}}^{\frac{1}{q}} &= \E{(HDx)_i^{2q}}^{\frac{1}{q}} = \Biggl(\mathrm{E}\biggl[\Bigl(\sum_{j=1}^m H_{ij} x_j D_{jj}\Bigr)^{2q}\biggr]^{\frac{1}{2q}}\Biggr)^2 \\
  &\lesssim q \cdot \sum_{j=1}^m H_{ij}^2 x_j^2 = q \cdot \sum_{j=1}^m x_j^2/m = q/m.\notag
\end{align}

  Then, we have
  \begin{align*}
    \frac{pm}{d} \cdot \bigl\|\max_i \eta_i z_i^2 \bigr\|_{L_p}
    & \leq \frac{pm}{d} \cdot \bigl\|\max_i \eta_i z_i^2 \bigr\|_{L_q} \\
    & = \frac{pm}{d} \cdot \mathbb{E}\Bigl[\max_i \eta_i z_i^{2q}\Bigr]^{1/q} \\
    & \leq \frac{pm}{d} \cdot \mathbb{E}\Bigl[\sum_i \eta_i z_i^{2q}\Bigr]^{1/q} \\
    & \leq \frac{pm}{d} \cdot \left(m \cdot \max_i \mathbb{E}\bigl[\eta_i z_i^{2q}\bigr]\right)^{1/q} \\
    & \lesssim \frac{pm}{d} \cdot \max_i \mathbb{E}\bigl[z_i^{2q}\bigr]^{1/q} \\
    & \lesssim \frac{pq}{d} \\
    & \leq \frac{p^2 + p \log d}{d},
  \end{align*}
  where the fifth step follows from identity
  $m \cdot \max_i \mathbb{E}\bigl[\eta_i z_i^{2q}\bigr] = d \cdot \max_i \mathbb{E}\bigl[z_i^{2q}\bigr]$
  and bound $q \geq \log d$, the sixth step
  follows from~(\ref{ineq:z-i}),
  and the seventh step follows from the definition of $q$.
\end{proof}

\begin{corollary}
  \label{fast_jl_tail}
  Fix $t > C_0 \sqrt{\frac{\log d}{d}}$, where $C_0$ is a sufficiently large constant. Then, for every unit vector $x \in \mathbb{R}^m$, we have
  \begin{enumerate}
  \item If $t \leq \frac{\log d}{\sqrt{d}}$,
    then
    $$
    \mathrm{Pr}\left[\left|\|\Pi x\|_2^2 - 1\right| \geq t\right] \leq e^{-\Omega\left(\frac{t^2 d}{\log d}\right)}
    $$
  \item If $\frac{\log d}{\sqrt{d}} \leq t \leq 1$,
    then
    $$
    \mathrm{Pr}\left[\left|\|\Pi x\|_2^2 - 1\right| \geq t\right] \leq e^{-\Omega\left(t \sqrt{d}\right)}
    $$
  \item If $t \geq 1$,
    then
    $$
    \mathrm{Pr}\left[\left|\|\Pi x\|_2^2 - 1\right| \geq t\right] \leq e^{-\Omega\left(\sqrt{td}\right)}
    $$
  \end{enumerate}
\end{corollary}
\begin{proof}
First, let us define $p^*$ so that
$$\bigl\|\|\Pi x\|_2^2 - 1\bigr\|_{L_{p^*}} \leq 0.9 t.$$
It follows from Theorem~\ref{fast_jl_moment_bound} that we can define $p^*$ as follows
$$
  p^* = c \begin{cases}
    \frac{t^2 d}{\log d},& \mbox{if $t \leq \frac{\log d}{\sqrt{d}}$;}\\
    t \sqrt{d},& \mbox{if $\frac{\log d}{\sqrt{d}} \leq t \leq 1$;}\\
    \sqrt{td},& \mbox{if $t \geq 1$,}
  \end{cases}
$$
where $c > 0$ is a sufficiently small constant. We let $C_0$ be large enough so that $p^* \geq 1$.

Then by Markov's inequality, we have,
$$
  \mathrm{Pr}\bigl[\bigl|\|\Pi x\|_2^2 - 1\bigr| \geq t\bigr] \leq 0.9^{p^*} \leq e^{-\Omega(p^*)}.
$$
Plugging in the formula for $p^*$, we get the desired bounds.
\end{proof}
\begin{corollary}
  \label{fast_jl_num_rows}
  For $0 < \varepsilon, \delta < 1/3$,
  the map $\pi_{m, d}$ is an $(\varepsilon, \delta)$-dimension reduction if
\begin{equation}\label{eq:bound-on-d}
  d \geq C_1 \, \frac{\log(\nicefrac{1}{\delta})\log(\nicefrac{1}{\delta\varepsilon})}{\varepsilon^2},
\end{equation}
for a sufficiently large absolute constant $C_1 > 0$.
\end{corollary}
\begin{proof}
  We show that for every unit vector $x$,
  $$P \equiv \mathrm{Pr}\left[\left|\|\Pi x\|_2^2 - 1\right| \geq \varepsilon\right] \leq \delta,$$
  and thus $\Pi$ is a $(\varepsilon, \delta)$-dimension reduction.
  To this end, we apply Corollary~\ref{fast_jl_tail} with $t = \varepsilon$. In order to apply it, we need to check that the condition  $\varepsilon \geq  C_0\sqrt{\frac{\log d}{d}}$ holds.

  We choose $C_1$ large enough so that, in particular, $C_1 \geq C_2^2$, where $C_2 = \max(C_0^2, 3)$. Then,
  (\ref{eq:bound-on-d}) implies that $d \geq d_0 \equiv  C_2^2 \frac{\log \nicefrac{1}{\varepsilon}}{\varepsilon^2}$.
  Since the function $x\mapsto \sqrt{\frac{\log x}{x}}$
  is decreasing for $x\geq e^2$ and $d_0 \geq e$, we have
  $$C_0\sqrt{\frac{\log d}{d}} \leq C_0 \sqrt{\frac{\log d_0}{d_0}} \leq \frac{C_0}{C_2} \varepsilon
  \sqrt{\frac{2\log \nicefrac{1}{\varepsilon} + 2\log C_2 + \log\log \nicefrac{1}{\varepsilon}}{\log \nicefrac{1}{\varepsilon}} } \leq \frac{C_0 \sqrt{3 + 2\log C_2}}{C_2}  \varepsilon \leq \varepsilon.$$
  Thus, we can indeed apply Corollary~\ref{fast_jl_tail}.

  It remains to verify that the upper bound that Corollary~\ref{fast_jl_tail} gives on $P$ is less than or equal to $\delta$. We note that either Case 1 or Case 2 in Corollary~\ref{fast_jl_tail} takes place, since $\varepsilon \leq 1/3 < 1$. Thus, it suffices to show that both upper bounds
  $e^{-\Omega\left(\frac{\varepsilon^2 d}{\log d}\right)}$ and $e^{-\Omega\left(\varepsilon \sqrt{d}\right)}$
  are less than $\delta$ (if we choose $C_1$ sufficiently large).

  Note that
  $$\frac{\varepsilon^2 d}{\log d} = \Omega\left(\frac{C_1\log \nicefrac{1}{\delta}\, (\log \nicefrac{1}{\varepsilon} + \log\nicefrac{1}{\delta})}{\log C_1 +
  \log \nicefrac{1}{\varepsilon} + \log\log \nicefrac{1}{\delta}}\right) \geq \Omega\left(\frac{C_1}{1+ \log C_1}\right) \log \nicefrac{1}{\delta}.$$
  Thus, by choosing $C_1$ large enough, we have that the upper bound $e^{-\Omega\left(\frac{\varepsilon^2 d}{\log d}\right)}$ on $P$ is less than $\delta$.
  We also have, $\varepsilon \sqrt{d} \geq \Omega(\sqrt{C_1}) \log\nicefrac{1}{\delta}$. Again, if $C_1$ is large enough, we have that the upper bound $e^{-\Omega\left(\varepsilon \sqrt{d}\right)}$ on $P$ is less than $\delta$.
\end{proof}

\begin{lemma}
  \label{fast_jl_p_tail}
Fix $p \geq 1$.
Assume that $\delta \leq \varepsilon$, $d \gg p^2$, and $C_1$ in (\ref{eq:bound-on-d}) is large enough.
For every unit vector $x \in \mathbb{R}^m$, we have
$$
\mathbb{E}\left[\ONE\{\|\Pi x\|_2 \geq 1 + \varepsilon\} \cdot \left(\|\Pi x\|_2^p - (1 + \varepsilon)^p\right)\right] \leq e^{-\Omega(\varepsilon \sqrt{d})}.
$$
\end{lemma}
\begin{proof}
The proof will be similar to that of Lemma~C.1. We have,
\begin{equation}\label{eq:bound-on-tail-FJL}
\mathbb{E}\left[\ONE\{\|\Pi x\|_2 \geq 1 + \varepsilon\} \cdot \left(\|\Pi x\|_2^p - (1 + \varepsilon)^p\right)\right]
 \leq \int_{\varepsilon}^{\infty} p (1 + t)^{p-1} \cdot \mathrm{Pr}[\|\Pi x\|_2 \geq 1 + t] \, dt.
\end{equation}
We bound the probability $\mathrm{Pr}[\|\Pi x\|_2 \geq 1 + t]$ for $t > \varepsilon$ using Corollary~\ref{fast_jl_tail} as follows. Since $\delta \leq \varepsilon$, we have
that $d \geq \frac{C_1 \log^2 \nicefrac{1}{\varepsilon}}{\varepsilon^2}$ and therefore
$$
\frac{\log d}{\sqrt{d}} \lesssim \frac{\varepsilon \log \nicefrac{1}{\varepsilon}}{\sqrt{C_1}\log \nicefrac{1}{\varepsilon}} =\frac{\varepsilon}{\sqrt{C_1}}.$$
Thus, if $C_1$ is sufficiently large, we have $\varepsilon > \frac{\log d}{\sqrt{d}}$; that is, either Case 2 or 3 in Corollary~\ref{fast_jl_tail} holds. We conclude that
$$\mathrm{Pr}[\|\Pi x\|_2 \geq 1 + t] \leq e^{-\Omega(\min(t,\sqrt{t})\cdot \sqrt{d})}.$$
Now let us go back to inequality~(\ref{eq:bound-on-tail-FJL}). We have,
\begin{align*}
\int_{\varepsilon}^{\infty} p (1 + t)^{p-1} &\cdot \mathrm{Pr}[\|\Pi x\|_2 \geq 1 + t] \, dt
     \leq \int_{\varepsilon}^{\infty} p (1 + t)^{p-1} \cdot e^{-\Omega(\min(t,\sqrt{t}) \cdot \sqrt{d})} \, dt\\
    & \leq p\int_{\varepsilon}^{\infty} \left((1 + t)^{p-1} \cdot e^{-\Omega(\min(t,\sqrt{t}) \cdot \sqrt{d})}\right)  e^{-\Omega(\min(t,\sqrt{t})\cdot \sqrt{d})} \, dt \\
    & \leq p\int_{\varepsilon}^{\infty} \left(e^{(p-1) \log (1+t)-\Omega(\min(t,\sqrt{t}) \cdot \sqrt{d})}\right)  e^{-\Omega(\min(t,\sqrt{t}) \cdot \sqrt{d})} \, dt
  \end{align*}
Observe that when $d \gg p$, the expression in the parentheses is less than 1. We get,
$$
\mathbb{E}\left[\ONE\{\|\Pi x\|_2 \geq 1 + \varepsilon\} \cdot \left(\|\Pi x\|_2^p - (1 + \varepsilon)^p\right)\right]\leq
p\int_{\varepsilon}^{\infty} e^{-\Omega(\min(t,\sqrt{t}) \cdot \sqrt{d})} \, dt.
$$
We simplify this bound using that $\int_1^{\infty} e^{-c\sqrt{t}}\,dt = \frac{2 (1 + c)e^{-c}}{c^2}$
and $\int_{\varepsilon}^{1} e^{-ct}\,dt = \frac{e^{-c \varepsilon }-e^{-c}}{c}$. We get
$$
\mathbb{E}\left[\ONE\{\|\Pi x\|_2 \geq 1 + \varepsilon\} \cdot \left(\|\Pi x\|_2^p - (1 + \varepsilon)^p\right)\right]\leq O\left(\frac{p}{\sqrt{d}}\right)  e^{-\Omega(\sqrt{d} \varepsilon)}
\leq e^{-\Omega(\varepsilon\sqrt{d})},$$
since $d \gg p^2$.
\end{proof}

Thus, if $d$ satisfies inequality~(\ref{eq:bound-on-d}),
$d\gg p^2$, and $\delta < \varepsilon$, then the map $\pi_{m, d}$
is a $(\varepsilon, \delta, \rho)$-dimension reduction with $\rho = e^{-\Omega(\varepsilon \sqrt{d})}$.
Finally, we conclude the proof of Theorem~\ref{fast_jl_main_thm}. We repeat the proof
of Theorem~\ref{thm:main} using the new bounds on $d$ and $\rho$.
Instead of~(\ref{eq:cond-on-d}), we obtain the following bound on $d$:
\begin{equation}
\label{new_d_bound}
d = C \cdot \left(\frac{\log^2 \frac{k}{\alpha} + p^2 \log^2 \frac{1}{\varepsilon} + p^4}{\varepsilon^2}\right),
\end{equation}
where $C$ is a sufficiently large absolute constant.
Indeed, using the notation from the proof of Theorem~\ref{thm:main}, we have:
\begin{equation}
  \label{theta_stupid_bound}
\theta \geq \frac{\alpha \cdot \varepsilon^{O(p)}}{e^{O(p^2)} \cdot k}.
\end{equation}
And we need to set $\delta$ and $d$ such that $\delta \ll \theta^7 / k$,
$\delta \ll \alpha / k^2$ and $\rho \leq \theta$.
We can set
$$
\delta = c \cdot \frac{\alpha^{O(1)} \cdot \varepsilon^{O(p)}}{e^{O(p^2)} \cdot k^{O(1)}},
$$
where $c > 0$ is a sufficiently small absolute constant,
such that the first two conditions hold. This together with Corollary~\ref{fast_jl_num_rows} implies that
as long as $d$ satisfies the following bound,
\begin{equation}
  \label{fast_jl_bound_calculation}
d \gtrsim \frac{\left(p^2 + \log(k / \alpha) + p \cdot \log(1 / \varepsilon)\right)^2}{\varepsilon^2}
\approx \frac{p^4 + \log^2(k / \alpha) + p^2 \cdot \log^2(1 / \varepsilon)}{\varepsilon^2}
\end{equation}
map $\pi_{m,d}$ is an $(\varepsilon, \delta)$-dimension reduction.
Now let us check that for this value of $d$, the map is also $(\varepsilon, \delta, \rho)$-dimension
reduction for $\rho \leq \theta$. Indeed, since $\delta \leq \varepsilon$ and $d \gg p^2$, we can use
Lemma~\ref{fast_jl_p_tail}, which implies that $\rho = e^{-\Omega(\varepsilon \sqrt{d})}$.
Combining this bound with~(\ref{theta_stupid_bound}), we can check that as long as $d$ satisfies~(\ref{fast_jl_bound_calculation}), map
$\pi_{m, d}$ is an $(\varepsilon, \delta, \rho)$-dimension reduction
for $\rho \leq \theta$.

Finally, we plug~(\ref{fast_jl_bound_calculation}) in the proof of Theorem~\ref{thm:main-alt}
and obtain the final bound~(\ref{fast_jl_bound}).